\documentclass{lmcs}
\pdfoutput=1

\usepackage{lastpage}
\lmcsdoi{16}{3}{2}
\lmcsheading{}{\pageref{LastPage}}{}{}%
{Feb.~05,~2019}{Jul.~03,~2020}{}

\usepackage[utf8]{inputenc}


\RequirePackage{doi}
\RequirePackage{hyperref}

\usepackage{amsmath,amssymb}
\usepackage{bbold}
\usepackage{mathtools}
\usepackage{etoolbox}
\usepackage{xifthen}

\usepackage{array}

\newcolumntype{L}[1]{>{\raggedright\let\newline\\\arraybackslash\hspace{0pt}}m{#1}}
\newcolumntype{C}[1]{>{\centering\let\newline\\\arraybackslash\hspace{0pt}}m{#1}}
\newcolumntype{R}[1]{>{\raggedleft\let\newline\\\arraybackslash\hspace{0pt}}m{#1}}

\usepackage{wrapfig}
\usepackage{graphicx}

\usepackage{pdflscape}
\usepackage{afterpage}

\usepackage{booktabs}

\usepackage{adjustbox}

\newcolumntype{D}[2]{%
    >{\adjustbox{angle=#1,lap=\width-(#2)}\bgroup}
    l%
    <{\egroup}%
}
\newcommand*\rot{\multicolumn{1}{D{45}{1em}}}

\newcommand{\YES}{\ensuremath{\checkmark}}

\newcommand{\MAYBE}{\ensuremath{?}}

\usepackage{tikz}
\usetikzlibrary{%
	automata,
	arrows,
	arrows.meta,
	cd,
	external,
	matrix,
	shadows,
	decorations.markings,
	decorations.pathmorphing,
	calc,shadows.blur,spy,
	positioning,
	shapes.arrows,
	shapes.geometric,
	intersections,
	tikzmark}

\pgfdeclarelayer{bg}    
\pgfdeclarelayer{aux}    
\pgfsetlayers{aux,bg,main}  

\makeatletter

\long\def\ifnodedefined#1#2#3{
  \@ifundefined{pgf@sh@ns@#1}{#3}{#2}}

\newcommand\aeundefinenode[1]{
  \expandafter\ifx\csname pgf@sh@ns@#1\endcsname\relax
  \else
    \global\expandafter\let\csname pgf@sh@ns@#1\endcsname\relax
  \fi
}

\newcommand\aeundefinethesenodes[1]{
  \foreach \myn  in {#1} 
    {
      \expandafter\aeundefinenode\expandafter{\myn}
    }
}

\makeatother

\tikzset{
	every picture/.append style={baseline={([yshift=-.5ex]current bounding box.center)},
	every loop/.style={}
	}
	}

\tikzset{
   commutative diagrams/.cd,
   arrow style=tikz,
   diagrams={>=latex}}

\tikzset{
    	lablRot/.style={
		anchor=center,
		rotate=#1,
		inner sep=.5mm},
	lablRot/.default=135,
	shortA/.style={
		shorten >=5
	},
	veryShortA/.style={
		shorten >=10
	},
	veryVeryShortA/.style={
		shorten >=3mm
	},
	partial base/.style={
		draw,
		densely dotted, thin
	},
	partialDel/.style={
		partial base,
		-{Rays[scale=0.5]},
	},
	partialId/.style={
		partial base,
	},
	partialCrea/.style={
		partial base,
		{Rays[scale=0.5]}-,
	},
	edge base/.style={
		draw=black!70,
		thick
	},
	undirEdge/.style={
		edge base
	},
	dirEdge/.style={
		edge base,
		decoration={
         	   markings,
         	   mark=at position 0.5 with {\arrow[scale=0.9,xshift=3.333pt]{latex}},
       		 },
      		postaction=decorate
	},
	otherDirEdge/.style={
		edge base,
		decoration={
         	   markings,
         	   mark=at position 0.5 with {\arrow[rotate=180,scale=0.9,xshift=3.333pt]{latex}},
       		 },
      		postaction=decorate
	},
	match/.style={
		draw=blue!40!black,
		thick,
		decoration={
         	   markings,
         	   mark=at position 0.5 with {\arrow[scale=0.9,xshift=3.333pt]{latex}},
       		 },
      		postaction=decorate
	},
	edgeMatch/.style={
		match,
		double
	},
	vertices/.style={
                draw,
                minimum width=1.5mm,
                outer sep=0pt,
                inner sep=0pt,
                align=center,
                circle,
                fill=black
  	},
	posNodes/.style={
		anchor=center,
		outer sep=0pt
	},
	, 
	sd/.style={
		thick
	},
	o/.style={
        shorten <=#1,
        decoration={
            markings,
            mark={
                at position 0
                with {
                    \draw circle [radius=#1]; 
                }
            }
        },
        postaction=decorate
    },
    o/.default=2pt,
    	oe/.style={
        shorten >=#1,
        decoration={
            markings,
            mark={
                at position 1
                with {
                    \draw circle [radius=#1]; 
                }
            }
        },
        postaction=decorate
    },
    oe/.default=2pt,
    , 
    	algMorph/.style={
		fill=white,
		regular polygon,
		regular polygon sides=3,
		shape border rotate=180,
		draw=black,thick,
		inner sep=0pt
	},
	coAlgMorph/.style={
		fill=white,
		regular polygon,
		regular polygon sides=3,
		draw=black,thick
	},
	vertexWL/.style={
		draw=black,
		thick
	},
	vId/.style={
		vertexWL
	},
	vC/.style={
		vertexWL,
		Rays-
	},
	vA/.style={
		vertexWL,
		-Rays
	},
	, 
	vCprob/.style={
		vertexWL,
		Rays[red]-
	},
	vAprob/.style={
		vertexWL,
		-Rays[red]
	},
	vCrep/.style={
		vertexWL,
		Rays[blue]-
	},
	vArep/.style={
		vertexWL,
		-Rays[blue]
	},
	vM/.style={ 
		vertexWL,
		draw=blue!70!black
	},
	edgeWL/.style={
		draw=black,
		thick
		},
	eId/.style={
		edgeWL,
		|-|,
		decorate,
		decoration={
			snake,
			pre length=1.4mm,
			post length=0.7mm,
			amplitude=1mm,
      			segment length=4mm}
	},
	eC/.style={
		edgeWL,
		Rays-|,
		decorate,
		decoration={
			snake,
			pre length=1.9mm,
			post length=0.4mm,
			amplitude=1mm,
      			segment length=4mm}
	},
	eA/.style={
		edgeWL,
		|-Rays,
		decorate,
		decoration={
			snake,
			pre length=1.4mm,
			post length=1.4mm,
			amplitude=1mm,
      			segment length=4mm}
	},
	eCunProb/.style={
		edgeWL,
		Rays[blue]-|,
		decorate,
		decoration={
			snake,
			pre length=1.9mm,
			post length=1.4mm,
			amplitude=1mm,
      			segment length=4mm}
	},
	eAunProb/.style={
		edgeWL,
		|-Rays[blue],
		decorate,
		decoration={
			snake,
			pre length=1.4mm,
			post length=1.4mm,
			amplitude=1mm,
      			segment length=4mm}
	},
	eM/.style={ 
		edgeWL,
		draw=blue!70!black,
		|-|
	},
	eEPL/.style={
		draw=black!70!white,
		very thin,
		dashed
	},
	mMN/.style={
		matrix of math nodes,
		ampersand replacement=\&,
		nodes in empty cells,
		row sep=1.5em,column sep=1.5em
		},
		auxWLnode/.style={
			draw,
			densely dotted,
			fill=white,
			regular polygon, regular polygon sides=4,
			inner sep=3pt
		}
}

\newcommand{\createNodeOnDemand}[5]{%
\pgfmathtruncatemacro{\xPos}{#1}
\pgfmathtruncatemacro{\yPos}{#2}
\pgfmathsetmacro{\X}{(\yPos-1)*\colSEP}
\pgfmathsetmacro{\Y}{-(\xPos-1)*\rowSEP} 

\ifnodedefined{m-\xPos-\yPos}{
\aeundefinenode{m-\xPos-\yPos}
}{
}
	\node[vertices,fill=#3] (m-\xPos-\yPos) at (\X,\Y) {};

\expandafter\ifstrequal\expandafter{#5}{t}{%
\node at ($(\X,\Y) +(0,\lblSEP)$) {{\tiny $#4$}};}{
\node at ($(\X,\Y) -(0,\lblSEP)$) {{\tiny $#4$}};} 

}

\newcommand{\vI}[6]{

\createNodeOnDemand{#1}{#2}{#3}{#4}{t}

\createNodeOnDemand{#1+1}{#2}{#5}{#6}{b}

\pgfmathtruncatemacro{\xPos}{#1}
\pgfmathtruncatemacro{\yPos}{#2}

    \begin{pgfonlayer}{bg}    
  	  \draw ($(m-\xPos-\yPos.center)+(0,-\rowSEP)$) edge[partialId] ($(m-\xPos-\yPos.center)+(0,0)$); 
    \end{pgfonlayer}

}

\newcommand{\lA}[5]{%

\pgfmathtruncatemacro{\xPos}{#1}
\pgfmathtruncatemacro{\yPos}{#2}

\ifstrequal{#5}{r}{
	 \begin{pgfonlayer}{bg}    
		\draw (m-\xPos-\yPos.center) 
		edge[bend right,out=-\loopIN,in=-90,undirEdge,draw=#3]  %
		($(m-\xPos-\yPos.center)+(0.4*\colSEP,0)$);
		\draw (m-\xPos-\yPos.center) 
		edge[bend left,out=\loopIN,in=90,undirEdge,draw=#3]  %
		($(m-\xPos-\yPos.center)+(0.4*\colSEP,0)$);
		\draw ($(m-\xPos-\yPos.center)+(0.4*\colSEP,0)$) 
		edge[partialDel] ($(m-\xPos-\yPos.center)+(0.4*\colSEP,0.4*\rowSEP)$);
		\node at ($(m-\xPos-\yPos.center)+(0.4*\colSEP,-\lblSEP)$) {{\tiny{$#4$}}}; 
	\end{pgfonlayer}}{
	\begin{pgfonlayer}{bg}    
		\draw (m-\xPos-\yPos.center) 
		edge[bend right,out=-\loopIN,in=-90,undirEdge,draw=#3]  %
		($(m-\xPos-\yPos.center)-(0.4*\colSEP,0)$);
		\draw (m-\xPos-\yPos.center) 
		edge[bend left,out=\loopIN,in=90,undirEdge,draw=#3]  %
		($(m-\xPos-\yPos.center)-(0.4*\colSEP,0)$);
		\draw ($(m-\xPos-\yPos.center)-(0.4*\colSEP,0)$) 
		edge[partialDel] ($(m-\xPos-\yPos.center)-(0.4*\colSEP,-0.4*\rowSEP)$);
		\node at ($(m-\xPos-\yPos.center)-(0.4*\colSEP,\lblSEP)$) {{\tiny{$#4$}}}; 
	\end{pgfonlayer}}
}

\newcommand{\lC}[5]{%

\pgfmathtruncatemacro{\xPos}{#1}
\pgfmathtruncatemacro{\yPos}{#2}

\ifstrequal{#5}{r}{
	 \begin{pgfonlayer}{bg}    
		\draw (m-\xPos-\yPos.center) 
		edge[bend right,out=-\loopIN,in=-90,undirEdge,draw=#3]  %
		($(m-\xPos-\yPos.center)+(0.4*\colSEP,0)$);
		\draw (m-\xPos-\yPos.center) 
		edge[bend left,out=\loopIN,in=90,undirEdge,draw=#3]  %
		($(m-\xPos-\yPos.center)+(0.4*\colSEP,0)$);
		\draw ($(m-\xPos-\yPos.center)+(0.4*\colSEP,-0.4*\rowSEP)$) 
		edge[partialCrea] ($(m-\xPos-\yPos.center)+(0.4*\colSEP,0)$);
		\node at ($(m-\xPos-\yPos.center)+(0.4*\colSEP,\lblSEP)$) {{\tiny{$#4$}}}; 
	\end{pgfonlayer}}{
	\begin{pgfonlayer}{bg}    
		\draw (m-\xPos-\yPos.center) 
		edge[bend right,out=-\loopIN,in=-90,undirEdge,draw=#3]  %
		($(m-\xPos-\yPos.center)-(0.4*\colSEP,0)$);
		\draw (m-\xPos-\yPos.center) 
		edge[bend left,out=\loopIN,in=90,undirEdge,draw=#3]  %
		($(m-\xPos-\yPos.center)-(0.4*\colSEP,0)$);
		\draw ($(m-\xPos-\yPos.center)+(-0.4*\colSEP,-0.4*\rowSEP)$) 
		edge[partialCrea] ($(m-\xPos-\yPos.center)-(0.4*\colSEP,0)$);
		\node at ($(m-\xPos-\yPos.center)+(-0.4*\colSEP,\lblSEP)$) {{\tiny{$#4$}}}; 
	\end{pgfonlayer}}
}

\newcommand{\lI}[7]{%

\pgfmathtruncatemacro{\xPos}{#1}
\pgfmathtruncatemacro{\yPos}{#2}

\ifstrequal{#7}{r}{
	  \begin{pgfonlayer}{bg}    
		\draw (m-\xPos-\yPos.center) 
		edge[bend right,out=-\loopIN,in=-90,undirEdge,draw=#3]  %
		($(m-\xPos-\yPos.center)+(0.4*\colSEP,0)$); %
		\draw (m-\xPos-\yPos.center) 
		edge[bend left,out=\loopIN,in=90,undirEdge,draw=#3]  %
		($(m-\xPos-\yPos.center)+(0.4*\colSEP,0)$);
	 \end{pgfonlayer}
	\begin{pgfonlayer}{aux}    
		\draw ($(m-\xPos-\yPos.center)-(0,\rowSEP)$) 
		edge[bend right,out=-\loopIN,in=-90,undirEdge,draw=#5]  %
		($(m-\xPos-\yPos.center)+(0.4*\colSEP,-\rowSEP)$); %
		\draw ($(m-\xPos-\yPos.center)-(0,\rowSEP)$) 
		edge[bend left,out=\loopIN,in=90,undirEdge,draw=#5]  %
		($(m-\xPos-\yPos.center)+(0.4*\colSEP,-\rowSEP)$); %
		\draw ($(m-\xPos-\yPos.center)+(0.4*\colSEP,0)$) 
		edge[partialId] ($(m-\xPos-\yPos.center)+(0.4*\colSEP,-\rowSEP)$);
		\node at ($(m-\xPos-\yPos.center)+(0.4*\colSEP,\lblSEP)$) {{\tiny{$#4$}}}; 
		\node at ($(m-\xPos-\yPos.center)+(0.4*\colSEP,-\rowSEP-\lblSEP)$) {{\tiny{$#6$}}}; 
	\end{pgfonlayer}}{
	\begin{pgfonlayer}{bg}    
		\draw (m-\xPos-\yPos.center) 
		edge[bend right,out=-\loopIN,in=-90,undirEdge,draw=#3]  %
		($(m-\xPos-\yPos.center)+(-0.4*\colSEP,0)$); %
		\draw (m-\xPos-\yPos.center) 
		edge[bend left,out=\loopIN,in=90,undirEdge,draw=#3]  %
		($(m-\xPos-\yPos.center)+(-0.4*\colSEP,0)$);
	 \end{pgfonlayer}
	\begin{pgfonlayer}{aux}    
		\draw ($(m-\xPos-\yPos.center)-(0,\rowSEP)$) 
		edge[bend right,out=-\loopIN,in=-90,undirEdge,draw=#5]  %
		($(m-\xPos-\yPos.center)+(-0.4*\colSEP,-\rowSEP)$); %
		\draw ($(m-\xPos-\yPos.center)-(0,\rowSEP)$) 
		edge[bend left,out=\loopIN,in=90,undirEdge,draw=#5]  %
		($(m-\xPos-\yPos.center)+(-0.4*\colSEP,-\rowSEP)$); %
		\draw ($(m-\xPos-\yPos.center)+(-0.4*\colSEP,0)$) 
		edge[partialId] ($(m-\xPos-\yPos.center)+(-0.4*\colSEP,-\rowSEP)$);
		\node at ($(m-\xPos-\yPos.center)+(-0.4*\colSEP,\lblSEP)$) {{\tiny{$#4$}}}; 
		\node at ($(m-\xPos-\yPos.center)+(-0.4*\colSEP,-\rowSEP-\lblSEP)$) {{\tiny{$#6$}}}; 
	\end{pgfonlayer}}
}

\newcommand{\eA}[5]{%

\pgfmathtruncatemacro{\xPos}{#1}
\pgfmathtruncatemacro{\yPos}{#2}

\ifstrequal{#3}{<}{\def\eType{otherDirEdge}}{}
\ifstrequal{#3}{=}{\def\eType{undirEdge}}{}
\ifstrequal{#3}{>}{\def\eType{dirEdge}}{}

\begin{pgfonlayer}{bg}    
	\draw ($(m-\xPos-\yPos.center)+(0,0)$) edge[\eType,draw=#4]  ($(m-\xPos-\yPos.center)+(\colSEP,0)$); 
	\draw ($(m-\xPos-\yPos.center)+(0.5*\colSEP,0)$) 
		edge[partialDel]  ($(m-\xPos-\yPos.center)+(0.5*\colSEP,0.4*\rowSEP)$);
\end{pgfonlayer}

\node at ($(m-\xPos-\yPos.center)+(0.5*\colSEP,-\lblSEP)$) {{\tiny{$#5$}}}; 
}

\newcommand{\eC}[5]{%

\pgfmathtruncatemacro{\xPos}{#1}
\pgfmathtruncatemacro{\yPos}{#2}

\ifstrequal{#3}{<}{\def\eType{otherDirEdge}}{}
\ifstrequal{#3}{=}{\def\eType{undirEdge}}{}
\ifstrequal{#3}{>}{\def\eType{dirEdge}}{}

\begin{pgfonlayer}{bg}    
	\draw ($(m-\xPos-\yPos.center)+(0,0)$) edge[\eType,draw=#4]  ($(m-\xPos-\yPos.center)+(\colSEP,0)$); 
	\draw ($(m-\xPos-\yPos.center)+(0.5*\colSEP,-0.4*\rowSEP)$) 
		edge[partialCrea]  ($(m-\xPos-\yPos.center)+(0.5*\colSEP,0)$);
\end{pgfonlayer}

\node at ($(m-\xPos-\yPos.center)+(0.5*\colSEP,\lblSEP)$) {{\tiny{$#5$}}}; 
}

\newcommand{\eI}[7]{%

\pgfmathtruncatemacro{\xPos}{#1}
\pgfmathtruncatemacro{\yPos}{#2}

\ifstrequal{#3}{<}{\def\eType{otherDirEdge}}{}
\ifstrequal{#3}{=}{\def\eType{undirEdge}}{}
\ifstrequal{#3}{>}{\def\eType{dirEdge}}{}

\begin{pgfonlayer}{aux}    
	\draw ($(m-\xPos-\yPos.center)+(0,-\rowSEP)$) 
	edge[\eType,draw=#6]  ($(m-\xPos-\yPos.center)+(\colSEP,-\rowSEP)$);
\end{pgfonlayer}

\begin{pgfonlayer}{bg}    
	\draw ($(m-\xPos-\yPos.center)+(0.5*\colSEP,-\rowSEP)$) 
		edge[partialId]  ($(m-\xPos-\yPos.center)+(0.5*\colSEP,0)$);
	\draw ($(m-\xPos-\yPos.center)+(0,0)$) edge[\eType,draw=#4]  ($(m-\xPos-\yPos.center)+(\colSEP,0)$); 
\end{pgfonlayer}

\node at ($(m-\xPos-\yPos.center)+(0.5*\colSEP,\lblSEP)$) {{\tiny{$#5$}}}; 
\node at ($(m-\xPos-\yPos.center)+(0.5*\colSEP,-\rowSEP-\lblSEP)$) {{\tiny{$#7$}}}; 
}

\newcommand{\eExtraA}[5]{%

\pgfmathtruncatemacro{\xPos}{#1}
\pgfmathtruncatemacro{\yPos}{#2}

\ifstrequal{#3}{<}{\def\eType{otherDirEdge}}{}
\ifstrequal{#3}{=}{\def\eType{undirEdge}}{}
\ifstrequal{#3}{>}{\def\eType{dirEdge}}{}

 \begin{pgfonlayer}{aux}    
	\draw ($(m-\xPos-\yPos.center)+(0,0)$) 
	edge[bend right,out=-50,in=210,\eType,draw=#4]  %
	($(m-\xPos-\yPos.center)+(0.7*\colSEP,-\xtraEdgeSEP*\rowSEP)$); %
	\draw ($(m-\xPos-\yPos.center)+(0.7*\colSEP,-\xtraEdgeSEP*\rowSEP)$) 
	edge[bend right,out=-35,in=-150,undirEdge,draw=#4]  %
	($(m-\xPos-\yPos.center)+(\colSEP,0)$); %
	\draw ($(m-\xPos-\yPos.center)+(0.7*\colSEP,-\xtraEdgeSEP*\rowSEP)$) 
	edge[partialDel] ($(m-\xPos-\yPos.center)+(0.7*\colSEP,0.4*\rowSEP)$);
 \end{pgfonlayer}

 	\node at ($(m-\xPos-\yPos.center)+(0.7*\colSEP,-\xtraEdgeSEP*\rowSEP-\lblSEP)$) {{\tiny$#5$}}; 

}

\newcommand{\eExtraC}[5]{%

\pgfmathtruncatemacro{\xPos}{#1}
\pgfmathtruncatemacro{\yPos}{#2}

\ifstrequal{#3}{<}{\def\eType{otherDirEdge}}{}
\ifstrequal{#3}{=}{\def\eType{undirEdge}}{}
\ifstrequal{#3}{>}{\def\eType{dirEdge}}{}

 \begin{pgfonlayer}{aux}    
	\draw ($(m-\xPos-\yPos.center)+(0,0)$) 
	edge[bend right,out=50,in=-210,\eType,draw=#4]  %
	($(m-\xPos-\yPos.center)+(0.7\colSEP,\xtraEdgeSEP*\rowSEP)$); %
	\draw ($(m-\xPos-\yPos.center)+(0.7\colSEP,\xtraEdgeSEP*\rowSEP)$) 
	edge[bend right,out=35,in=150,undirEdge,draw=#4]  %
	($(m-\xPos-\yPos.center)+(\colSEP,0)$); %
	\draw ($(m-\xPos-\yPos.center)+(0.7*\colSEP,-0.4*\rowSEP)$) 
	edge[partialCrea] ($(m-\xPos-\yPos.center)+(0.7*\colSEP,\xtraEdgeSEP*\rowSEP)$);
 \end{pgfonlayer}

 	\node at ($(m-\xPos-\yPos.center)+(0.7*\colSEP,\xtraEdgeSEP*\rowSEP+\lblSEP)$) {{\tiny$#5$}}; 

}

\newcommand{\eExtraI}[7]{%

\pgfmathtruncatemacro{\xPos}{#1}
\pgfmathtruncatemacro{\yPos}{#2}

\ifstrequal{#3}{<}{\def\eType{otherDirEdge}}{}
\ifstrequal{#3}{=}{\def\eType{undirEdge}}{}
\ifstrequal{#3}{>}{\def\eType{dirEdge}}{}

 \begin{pgfonlayer}{bg}    
	\draw ($(m-\xPos-\yPos.center)+(0,0)$) 
	edge[bend right,out=50,in=-210,\eType,draw=#4]  %
	($(m-\xPos-\yPos.center)+(0.7*\colSEP,\xtraEdgeSEP*\rowSEP)$); %
	\draw ($(m-\xPos-\yPos.center)+(0.7*\colSEP,\xtraEdgeSEP*\rowSEP)$) 
	edge[bend right,out=35,in=150,undirEdge,draw=#4]  %
	($(m-\xPos-\yPos.center)+(\colSEP,0)$); %
 \end{pgfonlayer}

 	\node at ($(m-\xPos-\yPos.center)+(0.7*\colSEP,\xtraEdgeSEP*\rowSEP+\lblSEP)$) {{\tiny$#5$}}; 

\pgfmathtruncatemacro{\xPos}{#1+1}
\pgfmathtruncatemacro{\yPos}{#2}

 \begin{pgfonlayer}{aux}    
	\draw ($(m-\xPos-\yPos.center)+(0,0)$) 
	edge[bend right,out=-50,in=210,\eType,draw=#6]  %
	($(m-\xPos-\yPos.center)+(0.7*\colSEP,-\xtraEdgeSEP*\rowSEP)$); %
	\draw ($(m-\xPos-\yPos.center)+(0.7*\colSEP,-\xtraEdgeSEP*\rowSEP)$) 
	edge[bend right,out=-35,in=-150,undirEdge,draw=#6]  %
	($(m-\xPos-\yPos.center)+(\colSEP,0)$); %
	\draw ($(m-\xPos-\yPos.center)+(0.7*\colSEP,-\xtraEdgeSEP*\rowSEP)$) 
	edge[partialId] ($(m-\xPos-\yPos.center)+(0.7*\colSEP,\xtraEdgeSEP*\rowSEP+\rowSEP)$);
 \end{pgfonlayer}

 	\node at ($(m-\xPos-\yPos.center)+(0.7*\colSEP,-\xtraEdgeSEP*\rowSEP-\lblSEP)$) {{\tiny$#7$}}; 

}

\gdef\tpScale{0.7}

\newcommand{\tP}[1]{\ensuremath{\vcenter{\hbox{\begin{tikzpicture}[transform shape, scale=\tpScale]\normalsize#1\end{tikzpicture}}}}} 

\newcommand{\tPL}[1]{\begin{tikzpicture}[transform shape, scale=\tpScale]\normalsize#1\end{tikzpicture}} 

\gdef\rowSEP{0.8}
\gdef\colSEP{0.8}
\gdef\loopIN{25} 
\gdef\lblSEP{0.2} 
\gdef\xtraEdgeSEP{0.3} 

\colorlet{h1color}{blue!70!black} 
\colorlet{h2color}{orange!90!black} 
\colorlet{h3color}{blue!40!white} 
\colorlet{h4color}{green!40!black} 

\usepackage{tikz-cd}

\makeatletter
\def\bstr{b}
\def\bfstr{bf}
\def\cstr{c}
\def\fstr{f}
\def\lst{A,B,C,D,d,E,F,G,H,I,J,K,L,M,N,O,P,Q,R,S,T,U,V,W,X,Y,Z}

\newcommand{\MkB}[1]{\expandafter\def\csname\bstr#1\endcsname{\mathbb{#1}}}
  \@for\i:=\lst\do{%
    \expandafter\MkB \i     } 

\newcommand{\MkBF}[1]{\expandafter\def\csname\bfstr#1\endcsname{\mathbf{#1}}}
  \@for\i:=\lst\do{%
    \expandafter\MkBF \i     } 

\newcommand{\MkCal}[1]{\expandafter\def\csname\cstr#1\endcsname{\mathcal{#1}}}
  \@for\i:=\lst\do{%
    \expandafter\MkCal \i     } 

\newcommand{\MkFrak}[1]{\expandafter\def\csname\fstr#1\endcsname{\mathfrak{#1}}}
  \@for\i:=\lst\do{%
    \expandafter\MkFrak \i     } 

\makeatother

\newcommand{\pB}[1]{\mathsf{PB}(#1)} 
\newcommand{\pO}[1]{\mathsf{PO}(#1)} 

\newcommand{\mono}[1]{\mathsf{mono}(#1)} 

\newcommand{\mor}[1]{\mathsf{mor}(#1)} 

\newcommand{\obj}[1]{\mathsf{obj}(#1)} 

\newcommand{\Lin}[1]{\mathsf{Lin}(#1)} 

\newcommand{\Prob}[1]{\mathsf{Prob}(#1)} 

\newcommand{\Stoch}[1]{\mathsf{Stoch}(#1)} 

\newcommand{\Match}[2]{\mathsf{M}_{#1}(#2)} 

\newcommand{\RMatch}[2]{\mathbf{M}_{#1}(#2)} 

\newcommand{\cSquare}[1]{\Box(#1)}

\newcommand{\comp}[3]{#1 \stackrel{#2}{\blacktriangleleft} #3}

\newcommand{\GRule}[3]{#1\xleftharpoonup{#2}#3}

\newcommand{\grule}[3]{#1\xLeftarrow{#2}#3}

\newcommand{\from}{\mathrel{:}}
\newcommand{\mspan}[1]{\mathsf{Span}_{\cM}(#1)} 

\newcommand{\bra}[1]{\left\langle#1\right\vert}
\newcommand{\ket}[1]{\left\vert#1\right\rangle}
\newcommand{\braket}[2]{\left\langle\left.#1\right\vert#2\right\rangle}

\newcommand{\mIO}{\mathbb{I}}

\newcommand{\OneVertG}[1][]{%
\ifthenelse{\equal{#1}{}}{%
\tP{%
\node[vertices] (a) at (1,1) {};}}{%
\tP{%
\node[vertices,#1] (a) at (1,1) {};}}}

\newcommand{\TwoVertG}[1][]{%
\ifthenelse{\equal{#1}{}}{%
\tP{%
\node[vertices] (a) at (1,1) {};
\node[vertices] (b) at (1.5,1) {};}}{%
\tP{%
\node[vertices,#1] (a) at (1,1) {};
\node[vertices,#1] (b) at (1.5,1) {};}}}

\newcommand{\TwoVertGL}[3][]{%
\ifthenelse{\equal{#1}{}}{%
\tP{%
\node[vertices] (a) at (1,1) {#2};
\node[vertices] (b) at (1.5,1) {#3};}}{%
\tP{%
\node[vertices,#1] (a) at (1,1) {#2};
\node[vertices,#1] (b) at (1.5,1) {#3};}}}

\newcommand{\TwoVertEdgeG}[1][]{%
\ifthenelse{\equal{#1}{}}{\tP{%
\node[vertices] (a) at (1,1) {};
\node[vertices] (b) at (1.5,1) {};
\draw (a) edge (b);}}{\tP{
\node[vertices,#1] (a) at (1,1) {};
\node[vertices,#1] (b) at (1.5,1) {};
\draw (a) edge[#1] (b);}}} 

\newcommand{\TwoVertEdgeGLb}[3]{%
\raisebox{-0.3em}{\tPL{%
\node[vertices,#1] (a) at (1,1) {};
\node[vertices,#1] (b) at (1.5,1) {};
\node at ($(a.south) -(0,\lblSEP)$) {{\tiny $#2$}}; 
\node at ($(b.south) -(0,\lblSEP)$) {{\tiny $#3$}}; 
\draw (a) edge[#1] (b);}}} 

\gdef\mycdScale{1}

\usepackage{environ}

\def\temp{&} \catcode`&=\active \let&=\temp 

\NewEnviron{mycd}[1][]{%
\begin{tikzpicture}[baseline=(current bounding box.west),
transform shape, scale=\mycdScale]
\node[inner sep=0, outer sep=0] {\begin{tikzcd}[#1]
\BODY
\end{tikzcd}};
\end{tikzpicture}} 

\begin{document}

\title[Rule Algebras for Adhesive Categories]{Rule Algebras for Adhesive Categories\rsuper*}
\titlecomment{{\lsuper*}A short version of this work has appeared in the Proceedings of the 27th EACSL Annual Conference on Computer Science Logic, CSL 2018~\cite{Behr2018}}

\author[N.~Behr]{Nicolas Behr\rsuper{a}}
\address{\lsuper{a}IRIF, Universit\'{e} Paris-Diderot, Paris, France}
\email[corresponding author]{nicolas.behr@irif.fr}
\thanks{The work of N.B.\ was supported by a \emph{Marie Sk\l{}odowska-Curie Individual Fellowship} (Grant Agreement No.~753750 --- RaSiR).}

\author[P.~Soboci\'{n}ski]{Pawe{\l} Soboci\'{n}ski\rsuper{b}}
\address{\lsuper{b}Tallinn University of Technology, Tallinn, Estonia}
\email{pawel@cs.ioc.ee}
\thanks{P.S.\ was supported by the ESF funded Estonian IT Academy research measure (project 2014-2020.4.05.19-0001).} 


\begin{abstract}
  \noindent We demonstrate that the most well-known approach to rewriting graphical structures, the Double-Pushout (DPO) approach, possesses a notion of sequential compositions of rules along an overlap that is associative in a natural sense. Notably, our results hold in the general setting of $\cM$-adhesive categories. This observation complements the classical Concurrency Theorem of DPO rewriting. We then proceed to define rule algebras in both settings, where the most general categories permissible are the finitary (or finitary restrictions of) $\cM$-adhesive categories with $\cM$-effective unions. If in addition a given such category possess an $\cM$-initial object, the resulting rule algebra is unital (in addition to being associative). We demonstrate that in this setting a canonical representation of the rule algebras is obtainable, which opens the possibility of applying the concept to define and compute the evolution of statistical moments of observables in stochastic DPO rewriting systems.
\end{abstract}

\maketitle

\section*{Introduction}\label{S:one}

Double pushout graph (DPO) rewriting~\cite{Ehrig1976} is the most well-known and influential approach to algebraic graph transformation. The rewriting mechanics are specified in terms of the universal properties of pushouts---for this reason, the approach is domain-independent and instantiates across a number of concrete notions of graphs and graph-like structures. Moreover, the introduction of adhesive, quasi-adhesive, weak adhesive and $\cM$-adhesive categories~\cite{lack2005adhesive,garner2012axioms,ehrig2004adhesive}---which, roughly speaking, ensure that the pushouts involved are ``well-behaved'', i.e.\ they satisfy similar exactness properties as pushouts in the category of sets and functions---entails that a standard corpus of theorems~\cite{DBLP:conf/gg/1997handbook} that ensures the ``good behaviour'' of DPO rewriting holds if the underlying ambient category is (quasi-, weak,\,$\cM$-)adhesive. 

An important classical theorem of DPO rewriting is the \emph{Concurrency Theorem}~\cite{ehrig:2006aa}, which involves an analysis of \emph{two} DPO productions applied in series. Given a \emph{dependency relation} (which, intuitively, determines how the right-hand side of the first rule overlaps with the left-hand side of the second), a purely category-theoretic construction results in a \emph{composite} rule which applies the two rules simultaneously. The Concurrency Theorem then states that the two rules can be applied in series in a way consistent with the relevant dependency relation if and only if the composite rule can be applied, yielding the same result.

The operation that takes two rules together with a dependency relation and produces a composite rule
can be considered as an algebraic operation on the set of DPO productions for a given category.
From this viewpoint, it is natural to ask whether this composition
operation is associative. It is remarkable that this appears to have been open until recently: an elementary proof of this, in the context of adhesive categories, was announced by us in the conference version~\cite{Behr2018} of this article.

In this extended version we:
\begin{itemize}[label=$\triangleright$]
\item generalise the associativity result to the setting of various notions of $\cM$-adhesive categories, giving a careful account of the precise technical conditions that are involved in the proof, which is given in its entirety here for the first time;
\item tie the proof of associativity to the classical Concurrency Theorem~\cite{ehrig:2006aa}, showing the relevant categorical constructions that are shared by the two results
\item give a more complete and detailed account of how the associativity theorem leads to the rule algebra framework, on which we elaborate below.
\end{itemize}

Indeed, associativity is advantageous for a number of reasons. In~\cite{bdg2016,bdgh2016}, the first author and his team developed the \emph{rule algebra} framework for a concrete notion of multigraphs. Inspired by a standard construction in mathematical physics, the operation of rule composition along a common interface yields an associative algebra: given a free vector space with basis the set of DPO rules, the product of the associative algebra takes two basis elements to a formal sum, over all possible dependency relations, of their compositions. This associative algebra is useful in applications, being the formal carrier of combinatorial information that underlies \emph{stochastic} interpretations of rewriting. The most famous example in mathematical physics is the Heisenberg-Weyl algebra~\cite{blasiak2010combinatorial,blasiak2011combinatorial}, which served as the starting point for~\cite{bdg2016}. Indeed,~\cite{bdg2016,bdgh2016} generalised the Heisenberg-Weyl construction from mere set rewriting to multigraph rewriting. Our work, since it is expressed abstractly in terms of $\cM$-adhesive categories, entails that the Heisenberg-Weyl and the DPO graph rewriting rule algebra can \emph{both} be seen as two instances of the same construction, expressed in abstract categorical terms.

\textbf{Structure of the paper.} The necessary categorical preliminaries are collected in Section~\ref{sec:mAdh}. Our main original results are collected
in Section~\ref{sec:acDPO}: following a brief recap of the DPO framework we first return to the classic Concurrency Theorem in Section~\ref{sec:concur}, then prove our main associativity result (Theorem~\ref{thm:assocDPO}) in Section~\ref{sec:assoc}. We devote Section~\ref{sec:ACDrd} to developing the rule algebra framework in the abstract setting, and proceed to give a number of applications:  Heisenberg-Weyl algebra in Section~\ref{sec:HW}, applications to combinatorics in Section~\ref{sec:RT} and stochastic mechanics in Section~\ref{sec:SM}. Our concluding remarks are in Section~\ref{sec:conclusion}.

\newpage

\section{Background: \texorpdfstring{$\cM$}{M}-adhesive categories}%
\label{sec:mAdh}

We briefly review standard material, following mostly~\cite{lack2005adhesive} (see~\cite{DBLP:conf/gg/CorradiniMREHL97,DBLP:conf/gg/1997handbook} for further references).

\medskip

\begin{defi}[Van Kampen (VK) squares~\cite{lack2005adhesive}]\label{def:VK}
In a category $\bfC$, a pushout square $\star$ (below left) is a \emph{Van Kampen (VK) square} whenever the following \emph{VK condition} holds: in every commutative cube over the pushout square such as the one below right in which the back and right faces are pullbacks, the top face is a pushout if and only if the front and left faces are pullbacks.
\begin{equation*}\gdef\mycdScale{0.85}
\vcenter{\hbox{\begin{mycd}
&
C
  \ar[dl, "n" description]
  \ar[dr, phantom, "\star"{sloped}]
& [-25pt]
&
A
  \ar[ll, "f" description]
  \ar[dl,  "m" description]\\
D
&
&
B
  \ar[ll, "g" description]
&\\
\end{mycd}}}\qquad\qquad
\vcenter{\hbox{\begin{mycd}
&
C'
  \ar[dl, "n'" description]
  \ar[dd,  near start, "c" description]
  & [-25pt]
&
A'
  \ar[ll, "f'" description]
  \ar[dd, "a" description]
  \ar[dl,  "m'" description]\\
D'
  \ar[dd, near start, "d" description] & &
B'
  \ar[ll,crossing over,  "g'" description] &\\
&
C
  \ar[dl, "n" description]
  \ar[dr, phantom, "\star"{sloped}] & &
A
  \ar[ll, "f" description]
  \ar[dl,  "m" description]\\
D &&
B
  \ar[ll, "g" description]
  \ar[uu,leftarrow,crossing over, near end, "b" description] &\\
\end{mycd}}}
\end{equation*}
As an important variant (cf.\ e.g.~\cite{Grochau_Azzi_2019}), given a class of morphisms $\cM\subset \mor{\bfC}$, we let a \emph{weak horizontal (weak vertical) $\cM$-VK square} be defined as a pushout square $\star$ where the VK condition is only required to hold for those commutative cubes where all horizontal (all vertical) morphisms are in $\cM$. Vertical $\cM$-VK squares are alternatively referred to as \emph{\textbf{$\mathbf{\cM}$-VK squares}}.
\end{defi}

Various notions of adhesive categories of importance to rewriting theories are known in the literature, with a number of different naming conventions.  We opt here to follow the traditional convention~\cite{ehrig2010categorical}.
\begin{defi}[Variants of adhesive categories]\label{def:adhc}
Let $\bfC$ be a category.
\begin{itemize}[label=$\triangleright$]
  \item $\mathbf{C}$ is an \textbf{\emph{adhesive category}}~\cite{lack2005adhesive} if
  \begin{enumerate}[(I)]
      \item $\bfC$ has pullbacks along arbitrary morphisms,
      \item $\bfC$ has pushouts along monomorphisms, and
      \item pushouts along monomorphisms are VK squares.
  \end{enumerate}
  \item Let $\cM\subset \mono{\bfC}$ be a class of monomorphisms.
  \begin{itemize}[label=$\triangleright$]
  \item $(\bfC,\cM)$ is an \textbf{\emph{adhesive HLR category}}~\cite{ehrig:2006aa,ehrig2010categorical} if
  \begin{enumerate}[(I')]
      \item $\bfC$ has pullbacks along $\cM$-morphisms,
      \item $\bfC$ has pushouts along $\cM$-morphisms,
      \item pushouts along $\cM$-morphisms are VK squares, and
      \item $\cM$ contains all isomorphisms and is stable under composition, pullback and pushout.
  \end{enumerate}
  \item If instead of axiom~(III') above pushouts along $\cM$-morphisms are only required to be horizontal VK squares (vertical weak VK squares), $(\bfC,\cM)$ is referred to~\cite{Grochau_Azzi_2019} as a \textbf{\emph{horizontal weak (vertical weak) adhesive HLR category}}. If $\cM$-pushouts are both horizontal \emph{and} vertical VK squares, $\bfC$ is called \textbf{\emph{weak adhesive HLR category}}.
  \end{itemize}
\end{itemize}
As proposed in~\cite{ehrig2010categorical}, we will alternatively refer to  vertical weak adhesive HLR categories simply as \textbf{\emph{$\mathbf{\cM}$-adhesive categories}}.
\end{defi}

\begin{table}[htpb]
\centering
\vspace{2em}
{\setlength{\extrarowheight}{5pt}
  \begin{tabular}{C{4.5cm}ccccccccc}
Category\newline(underlying data type)\newline&
\rot{$\cM=\mathsf{mono}(\mathbf{C})$} &
 \rot{adhesive} &
\rot{adhesive HLR} &
\rot{horizontal weak adh.\ HLR} &
\rot{vertical weak adh.\ HLR}&
\rot{$\mathcal{M}$-initial object}&
\rot{$\mathcal{M}$-effective unions}&
\rot{references}
    \\[-0.5em] \toprule
$\mathbf{Set}$\newline (sets)      &
    \YES      &   \YES    &  \YES &
    \YES & \YES  & \YES & \YES & \cite{lack2005adhesive} \\ 
$\mathbf{Graph}$\newline(\emph{directed} multigraphs) &
    \YES      &   \YES    &  \YES &
    \YES & \YES  & \YES & \YES & \cite{lack2005adhesive}\\ 
$\hat{\mathbf{S}}$\newline(presheaves on category $\mathbf{S}$) &
    \YES      &   \YES    &  \YES &
    \YES & \YES  & \YES & \YES & \cite{Lack2006,Grochau_Azzi_2019}\\ 
$\mathbf{HyperGraph}$\newline(hypergraphs) &
    \YES      &       &  \YES &
    \YES & \YES  & \YES & \YES & \cite{ehrig:2006aa}\\ 
$\mathbf{AGraph}_{\Sigma}$\newline(attributed graphs over signature $\Sigma$) &
         &       &  \YES &
    \YES & \YES &  & \YES & \cite{ehrig:2006aa,Braatz:2010aa,Grochau_Azzi_2019}\\ 
$\mathbf{SymbGraph}_{D}$\newline(symbolic graphs over $\Sigma$-algebra $D$) &
         &       &  \YES &
    \YES & \YES  &  & \YES & \cite{orejas2010symbolic,Grochau_Azzi_2019}\\ 
$\mathbf{uGraph}$\newline(finite \emph{undirected} multigraphs) &
    \YES      &       &  &
    \YES & \YES  & \YES & \YES & $\begin{array}{c}\text{\cite{padberg2017towards}}\\\text{Section~\ref{sec:RT}}\end{array}$\\ 
$\mathbf{PTnets}$\newline(place/transition nets) &
          &       &   &
    \YES & \YES  & \YES & \YES & \cite{ehrig:2006aa,Braatz:2010aa}\\ 
$\mathbf{ElemNets}$\newline(elementary Petri nets) &
          &       &   &
    \YES & \YES  & \YES & \YES & \cite{ehrig:2006aa,Braatz:2010aa}\\ 
$\mathbf{Spec}$\newline(algebraic specifications) &
          &       &   &
    \YES & \YES  & \MAYBE & \MAYBE & \cite{ehrig:2006aa}\\ 
$\mathbf{lSets}$\newline(list sets) &
          &       &   &
     & \YES & \MAYBE  & \MAYBE & \cite{Heindel_2010}\\[1.5em] 
\bottomrule
\end{tabular}
}
\caption{\label{tab:adh}Examples of categories exhibiting various forms of adhesivity, with additional properties relevant to this paper and references for further technical details provided. The symbol \MAYBE~indicates when a certain property is (to the best of our knowledge) not known to hold.}
\vspace{2em}
\end{table}

\begin{exa}
  In order to illustrate the utility of the various different notions of adhesive categories, we list in Table~\ref{tab:adh} examples for each of these types quoted from the rewriting literature. The table lists additional properties of relevance to the present paper that will be subsequently introduced below. The construction principles used for many of the more sophisticated examples listed in the table are rooted in notions of slice, coslice, functor and comma categories (cf.\ e.g.~\cite{ehrig:2006aa,Braatz:2010aa} for further details). As an example of a comma category construction, we introduce in Section~\ref{sec:RT} the category $\mathbf{uGraph}$ of finite undirected multigraphs.
\end{exa}

\begin{rem}
We would like to emphasise the hierarchy among concepts as implied by Definition~\ref{def:adhc}, in that every adhesive category is an adhesive HLR category (for $\cM=\mono{\bfC}$), while every adhesive HLR category is a weak adhesive HLR category. By definition, weak adhesive HLR categories are both horizontal and vertical weak adhesive HLR categories. As discussed in detail in~~\cite{ehrig2010categorical}, horizontal weak adhesive HLR categories lack certain properties relevant to rewriting, in contrast to the vertical weak variant, which is why the latter is typically considered as the most general type of rewriting-compatible category. Henceforth, we will thus follow the traditional convention to refer to a vertical weak adhesive category as an $\cM$-adhesive category.
\end{rem}

In many applications of interest, the data structures to be rewritten satisfy a certain notion of finiteness:

\begin{defi}[{Finitary $\cM$-adhesive categories, cf.~\cite[Defs.~2.8 and~4.1]{Braatz:2010aa}}]
Let $(\bfC,\cM)$ be an $\cM$-adhesive category. $(\bfC,\cM)$ is called \emph{\textbf{finitary}} if each object is \emph{finite} (i.e.\ has only finitely many $\cM$-subobjects). The \emph{finitary restriction} $(\bfC_{{\rm fin}},\cM_{{\rm fin}})$ of $(\bfC,\cM)$ is defined as the restriction of $\bfC$ to the full subcategory $\bfC_{{\rm fin}}$ of finite objects, and with $\cM_{{\rm fin}}:=\cM\cap \bfC_{{\rm fin}}$.
\end{defi}

The following result guarantees that every $\cM$-adhesive category will give rise to a finitary variant via a form of restriction (i.e.\ which in particular preserves the adhesivity properties).

\begin{thmC}[{\cite[Thm.~4.6]{Braatz:2010aa}}]\label{thm:finRes}
The finitary restriction $(\bfC_{{\rm fin}},\cM_{{\rm fin}})$ of an $\cM$-adhesive category $(\bfC,\cM)$ is a finitary $\cM$-adhesive category.
\end{thmC}

An important concept used throughout the paper is that of (isomorphism classes of) spans of $\cM$-morphisms. A span consists of two $\cM$-morphisms with a common source $C\xhookleftarrow{b}B\xhookrightarrow{a}A$. A homomorphism of spans from
$C\xhookleftarrow{b}B\xhookrightarrow{a}A$ to $C\xhookleftarrow{b'}B'\xhookrightarrow{a'}A$ consists of a morphism $h\from B\rightarrow B'$ such
that the two resulting triangles commute. The spans are said to be isomorphic when $h$ is an isomorphism. The following shows that isomorphism classes of spans are composable, and so are the arrows of a category $\mspan{\bfC}$
with the same objects as $\bfC$.
\begin{lem}\label{lem:spans}
  Let $(\bfC,\cM)$ be an $\cM$-adhesive category, and let $R= (C\xhookleftarrow{b}B\xhookrightarrow{a}A)$ and $S=(E\xhookleftarrow{d}D\xhookrightarrow{c}C)$ be two composable spans with $a,b,c,d\in \cM$. Then their \emph{composition} $S\circ R$, calculated via taking the pullback marked $\mathsf{PB}$ below,
  \begin{equation}\label{eq:PBC}
  \begin{mycd}
    & &
    F\ar[dl,"f"']\ar[dr,"e"]\ar[dd,phantom,"\mathsf{PB}"] & &\\
    & D \ar[dl,"d"'] \ar[dr,"c"] & &
    B \ar[dl,"b"']\ar[dr,"a"] & \\
    E & & C & & A
  \end{mycd}\qquad S\circ R:= (E\xleftarrow{d\circ f}F\xrightarrow{a\circ e}A)\,,
  \end{equation}
  is also a span of $\cM$-morphisms (i.e. $d\circ f,a\circ e\in \cM$).
  \begin{proof}
    The proof follows from the $\cM$-adhesivity properties, i.e.\ from stability of the class $\cM$ under pullback and composition.
  \end{proof}
\end{lem}

Note in particular that the pullback composition operation $\circ$ on spans behaves in complete analogy to the composition operations of functions as well as on linear operators, at least if considering the following convention\footnote{This convention is standard in much of the mathematics literature; however, traditionally the opposite convention of reading spans ``left-to-right'' is encountered in the literature on graph rewriting. Since in our framework we will eventually assign linear operators to spans, the ``right-to-left'' convention offers the more convenient encoding.}:

\begin{conv}\label{nc}
  We read spans in the \textbf{\emph{``right-to-left'' convention}}, such that if we consider spans $R,S$ as above to encode partial functions $r\from A\rightharpoonup C,\, s\from C\rightharpoonup E$, then function composition and span composition are compatible (i.e.\ $s\circ r$ is computed via $S\circ R$).
\end{conv}

\subsection{Some useful technical results}

We recall first some basic pasting properties of pushouts and pullbacks that hold in any category.
\begin{lem}%
\label{lem:pushoutpullback}
Given a commutative diagram
\begin{equation*}\gdef\mycdScale{0.85}
\begin{mycd}
A\ar[r]\ar[d] & B \ar[r]\ar[d] &E\ar[d]\\
C\ar[r] & D\ar[r] & F\\
\end{mycd}\,,
\end{equation*}
\begin{itemize}[label=$\triangleright$]
\item (pullback version) if the right square is a pullback then the left square is a pullback if and only if the entire exterior rectangle is a pullback;
\item (pushout version) if the left square is a pushout then the right square is a pushout if and only if the entire exterior rectangle is a pushout.
\end{itemize}
\end{lem}

\begin{lem}\label{lem:auxPOPBcat}
In any category, given commutative diagrams of the form
 \begin{equation}
    \begin{mycd}
      A \ar[r,"f"] \ar[d,equal]\ar[dr,phantom,"(A)"] & B \ar[d,equal]\\
      A \ar[r,"f"'] & B
    \end{mycd}\qquad \begin{mycd}
      A \ar[r,equal]
      \ar[d,equal]\ar[dr,phantom,"(B)"] & A \ar[d,"g"]\\
      A \ar[r,"g"'] & B
    \end{mycd}\qquad \begin{mycd}
      A \ar[r,"f"]
      \ar[d,equal]\ar[dr,phantom,"(C)"] &
      B \ar[d,"g"]\\
      A \ar[r,"g\circ f"'] & C
    \end{mycd}\,,
  \end{equation}
it holds that
\begin{enumerate}[(I)]
\item the square marked $(A)$ is a pushout for arbitrary morphisms $f$,
\item the square marked $(B)$ is a pullback if and only if the morphism $g$ is a monomorphism
\item the square marked $(C)$ is a pullback for arbitrary morphisms $f$ if $g$ is a monomorphism.
\end{enumerate}
In addition, if the category is an $\cM$-adhesive category, statements (I), (III) and the ``if'' part of (II) hold for ``monomorphisms'' replaced by ``$\cM$-morphisms''.
\end{lem}
\begin{proof}
The statements $(I)$ and $(II)$ are classical, whence their proof is omitted for brevity. In order to prove the statement $(III)$, it suffices to combine $(I)$ and $(II)$ to conclude that a square $(C)$ as in the diagram below left
\begin{equation}
  \begin{mycd}
      A \ar[r,"f"]
      \ar[d,equal]\ar[dr,phantom,"(C)"] &
      B \ar[d,"g"]\\
      A \ar[r,"g\circ f"'] & C
    \end{mycd}\qquad \begin{mycd}
  A \ar[r,"f"'] \ar[rr, bend left, "f"] \ar[d,equal]
  \ar[dr,phantom,"(D)"]&
  B \ar[r,equal] \ar[d,equal]
  \ar[dr,phantom,"(E)"] &
  B \ar[d,"g"]\\
  A \ar[r,"f"] \ar[rr,bend right, "g\circ f"'] &
  B \ar[r,"g"] & C
\end{mycd}
\end{equation}
is a pullback square if $f$ is an arbitrary morphism and $g$ a monomorphism, since the square $(C)$ may be obtained as the composition of a pushout square $(D)$ along an isomorphism (whence a monomorphism), which is according to Lemma~\ref{lem:convenient} also a pullback, and a pullback square $(E)$, by pullback composition (Lemma~\ref{lem:pushoutpullback}) $(D)+(E)$ and thus $(C)$ is a pullback. As for the specialisation of the statements to the setting of $\cM$-adhesive categories, the claims follow trivially for (I) (no modification), and also for (III) and the ``if'' part of (II), since $\cM$ is assumed to be a class of monomorphisms.
\end{proof}

Next, we recall a number of useful properties of pushouts and pushout complements in $\cM$-adhesive categories.
\begin{lemC}[{\cite[Lemma~2.6]{EHRIG:2014ma}}]%
\label{lem:convenient}
In any $\cM$-adhesive category:
\begin{enumerate}[(I)]
\item Pushouts along $\cM$-monomorphisms are also pullbacks.
\item ($\cM$-pushout-pullback decomposition) if, in the following diagram
\begin{equation}\label{lem:poPbDec}
\vcenter{\hbox{\begin{mycd}
A
  \ar[r,"b"']
  \ar[d,"c"']
  \ar[rr, bend left, "="'] &
B
  \ar[r,"e"']\ar[d,"a"'] &
E
  \ar[d]\\
C
  \ar[r,"d"]
  \ar[rr, bend right, "="] &
D
  \ar[r,"f"] & F
\end{mycd}}}
\end{equation}
the exterior face is a pushout, the right face is a pullback, and $f\in \cM$ and ($b\in \cM$ or $c\in \cM$), then the left and right squares are both pushouts and pullbacks.
\item (uniqueness of pushout complements) given $A\hookrightarrow C$ in $\cM$ and $C\rightarrow D$, the respective pushout complement $A\xrightarrow{b} B \xhookrightarrow{a} D$ (if it exists) is unique up to isomorphism, and with $b\in \cM$ (due to stability of $\cM$-morphisms under pushouts).
\end{enumerate}
\end{lemC}

\noindent
Note that in~\eqref{lem:poPbDec} by virtue of stability of $\cM$-morphisms under pushout and pullback, these conditions entail that since $f\in \cM$, we also have that $e\in \cM$, while $b\in \cM$ means that $d\in \cM$ (and $c\in \cM$ that $a\in \cM$).

\subsection{Additional category-theoretical prerequisites}

Passing from adhesive categories to $\cM$-adhesive categories on the one hand permits to study rewriting in the most general setting for DPO rewriting known to date, yet it comes at the price of a number of technicalities that are necessary to ensure certain associativity properties for the rewriting as introduced in the main part of the paper. The first such property concerns the existence of the analogue of the empty set in the category $\mathbf{Set}$ or the empty graph in the category $\mathbf{Graph}$, referred to as \emph{$\cM$-initial object} for a general $\cM$-adhesive category (the existence of which is not guaranteed by $\cM$-adhesivity, cf.\ Table~\ref{tab:adh}). The second requirement concerns the property of an $\cM$-adhesive category possessing \emph{$\cM$-effective unions}, analogous (to a certain extent) to the notion of union of sets and related properties.

\begin{defi}[{$\cM$-initial object;~\cite[Def.~2.5]{Braatz:2010aa}}]\label{def:Minit}
An object $\mIO$ of an $\cM$-adhesive category $(\bfC,\cM)$ is defined to be an \emph{$\cM$-initial object} if for each object $A\in \obj{\bfC}$ there exists a unique monomorphism $i_A:\mIO\hookrightarrow A$, which is moreover required to be in $\cM$.
\end{defi}

\begin{lemC}[{\cite[Fact~2.6]{Braatz:2010aa}}]\label{lem:binaryCoproducts}
  If an $\cM$-adhesive category $(\bfC,\cM)$ possesses an $\cM$-initial object $\mIO\in \obj{\bfC}$, the category has \emph{finite coproducts}, and the coproduct injections are in $\cM$.
\end{lemC}
\begin{proof}
  We quote the proof from~\cite{Braatz:2010aa} for illustration of this important property: it suffices to consider the case of binary coproducts. One may construct the coproduct $A+B$ of two objects $A,B\in \obj{\bfC}$ via taking the pushout
  \begin{equation}
    \begin{mycd}
        & \mIO\ar[dl,"i_A"']\ar[dr,"i_B"]\ar[dd,phantom,"\mathsf{PO}"] &\\
        A\ar[dr,"in_A"'] & & B\ar[dl,"in_B"]\\
        & A+B
    \end{mycd}\,.
  \end{equation}
  Since the underlying category is assumed to be $\cM$-adhesive, according to Definition~\ref{def:adhc} the above pushout is guaranteed to exist since $i_A,i_B\in \cM$ via the assumption of $\mIO$ being an $\cM$-initial object, and by virtue of stability of $\cM$-morphisms under pushouts, we may moreover conclude that indeed $in_A,in_B\in \cM$.
\end{proof}

The second main property required for our construction of rule algebras concerns a certain compatibility property relating pushouts and pullbacks along $\cM$-morphisms:
\begin{defi}[$\cM$-effective unions]
  An $\cM$-adhesive category $(\bfC,\cM)$ is said to possess \emph{$\cM$-effective unions} if the following property holds: given a commutative diagram as below with $b_1,b_2,c_1,c_2,d_1,d_2\in\cM$, where $(1)$ is a pushout, the outer square a pullback, and where $x$ is the unique morphism induced by the universal property of the pushout,
\begin{equation}\label{eq:eu}
\begin{mycd}
A
  \ar[d,hook,"b_2"']\ar[r,hook,"b_1"]\ar[dr,phantom,"(1)"] & [-5pt]
B_1
  \ar[d,hook,"c_1"] \ar[ddr,hook,bend left,"d_1"] &[-15pt] \\
B_2
  \ar[r,hook,"c_2"']\ar[drr,hook,bend right,"d_2"'] & D
  \ar[dr,near start,"x"]&\\[-15pt]
  & & E
\end{mycd}
\end{equation}
then $x\in \cM$.
\end{defi}

As indicated in Table~\ref{tab:adh}, while the property of $\cM$-effective unions is traditionally well-known in a range of important examples, it is nonetheless in general a difficult task to establish this property. Many of the positive examples may be derived via the following (variant of a) result of~\cite{lack2005adhesive}:
\begin{thm}[{Variant of~\cite[Thm.~5.1]{lack2005adhesive}}]\label{thm:euAux}
  Let $(\bfC,\cM)$ be a horizontal weak adhesive HLR category. Then in a commutative diagram of the form~\eqref{eq:eu}, $x\in \mono{\bfC}$ is a monomorphism \emph{(not necessarily in $\cM$ if $\cM\neq\mono{\bfC}$)}.
\end{thm}
\begin{proof}
It may be easily verified that the proof provided in~\cite{lack2005adhesive} in the setting of adhesive category in fact utilises only the axioms of horizontal weak adhesive HLR categories, and thus directly generalises to this setting.
\end{proof}
As the inspection of Table~\ref{tab:adh} reveals, $\mathbf{lSet}$ is the only known example in which the horizontal, but not also the vertical weak adhesive HLR properties hold, while simultaneously in this category $\cM_{\mathbf{lSet}}\neq \mono{\mathbf{lSet}}$. Thus Theorem~\ref{thm:euAux} effectively identifies weak adhesive HLR categories with $\cM=\mono{\bfC}$ as natural examples of categories with ($\mono{\bfC}$)-effective unions. This includes the original statement that all adhesive categories have effective unions, but also covers examples of (weak) adhesive HLR categories such as $\mathbf{HyperGraph}$ and $\mathbf{uGraph}$. As recently demonstrated in~\cite{Grochau_Azzi_2019}, for $\cM$-adhesive categories such as $\mathbf{AGraph}_{\Sigma}$ or $\mathbf{SymbGraph}_{D}$ where the class $\cM$ of monomorphisms has additional structure, it is possible to use the statement of Theorem~\ref{thm:euAux} in order to manually prove the existence of $\cM$-effective unions. On the other hand, the latter two categories fail to possess an $\cM$-initial object, which prevents them from satisfying all assumptions necessary in order to support a unital rule algebra construction (albeit they \emph{do} support the concurrency and associativty property of DPO rule compositions). We leave a more in-depth investigation into these matters of admissibility for future work.

\section{Double-pushout (DPO) rewriting}\label{sec:acDPO}

We now recall  \emph{Double-Pushout (DPO) rewriting} for $\cM$-adhesive categories (adapted according to the results of~\cite{Braatz:2010aa} and to our \emph{notational convention}~\ref{nc}).

\begin{asm}\label{as:cats}
Throughout the remainder of this paper, we fix in each definition an \emph{$\cM$-adhesive category} $(\bfC,\cM)$ (typically just written as $\bfC$ for brevity) that is assumed to possess an $\cM$-initial object and $\cM$-effective unions.
\end{asm}

\begin{defiC}[{\cite[Def.~7.1]{lack2005adhesive}}]\label{def:prod}
A span $p$ of morphisms (with $O$utput, $K$ontext, $I$nput)
\begin{equation}
\label{eq:prod}
O\xleftarrow{o} K\xrightarrow{i} I
\end{equation}
is called a \emph{production}. $p$ is said to be 
\emph{linear} if both $i$ and $o$ are monomorphisms in $\cM$.  We denote the \emph{set of linear productions} by $\Lin{\bfC}$. We will also frequently make use of the \emph{alternative notation} $\GRule{O}{p}{I}$ where $p=(O\xleftarrow{o} K\xrightarrow{i} I)\in \Lin{\bfC}$.
\end{defiC}
A \emph{homomorphism of productions} $p\rightarrow p'$ consists of arrows, $O \rightarrow O'$, $K \rightarrow K'$ and $I\rightarrow I'$, such that the obvious diagram commutes. A homomorphism is an isomorphism when all of its components are isomorphisms. We do not distinguish between isomorphic productions. Note that the notion of morphism of productions is different than that for general spans.

\begin{defiC}[{\cite[Def.~7.2]{lack2005adhesive}}]\label{def:gc}
Given a production $p$ as in~\eqref{eq:prod}, a \emph{match} of $p$ in an object $X\in \obj{\bfC}$ is a morphism $m:I\rightarrow X$. A match is said to satisfy the \emph{gluing condition} if there exists an object $E$ and morphisms $k:K\rightarrow \overline{K}$ and $x:\overline{K}\rightarrow X$ such that~\eqref{eq:gluing} is a pushout.
\begin{equation}\gdef\mycdScale{0.85}
\label{eq:gluing}
\begin{mycd}
K\ar[d,"k"',dashed]\ar[dr,phantom,"\mathsf{PO}"] &
I \ar[l,leftarrow,"i"']\ar[d,"m"]\\
\overline{K} & X\ar[l,leftarrow,dashed,"x"]
\end{mycd}\,.
\end{equation}
More concisely, the \emph{gluing condition} holds if there is a \emph{pushout complement} of $K\xrightarrow{i}I\xrightarrow{m}X$.
\end{defiC}

From here on, we will focus solely on \emph{linear productions} and matches that are $\cM$-morphisms, which entails due to the above statements a number of practical simplifications, and which allows us to simplify also the notations as follows:
\begin{conv}
Unless mentioned otherwise, henceforward \emph{all} arrows are understood to be morphisms of the class $\cM$ of the underlying $\cM$-adhesive category $\bfC$, whence we will use the notation $\rightarrow$ of ``ordinary'' arrows (instead of $\hookrightarrow$) to denote arrows of $\cM$ in all diagrams and formulae.
\end{conv}

\begin{defi}[{compare~\cite[Def.~7.3]{lack2005adhesive}}]%
\label{def:DPOr}
Given an object $X\in \obj{\bfC}$ and a linear production $p\in \Lin{\bfC}$, we define the \emph{set of admissible matches} $\Match{p}{X}$ as the set of monomorphisms $m:I\rightarrow X$ in $\cM$ for which $m$ satisfies the \emph{gluing condition}. As a consequence, there exist objects and morphisms such that in the diagram below both squares are pushouts (where the square marked $\mathsf{POC}$ is constructed as a pushout complement):
\begin{equation}\label{eq:DPOr}\gdef\mycdScale{0.85}
\begin{mycd}
O \ar[d,"{m^{*}}"'] &
 K \ar[l,"o"']\ar[r,"i"]\ar[d,"k"']
    \ar[dl,phantom, "{\mathsf{PO}}"]\ar[dr,phantom,"{\mathsf{POC}}"] &
 I \ar[d,"m"]
 \\
 {X'} & {\overline{K}} \ar[l,"o'"]\ar[r,"i'"'] & X\\
\end{mycd}
\end{equation}
We write $p_m(X):=X'$ for the object ``produced'' by the above diagram. The process is called \emph{derivation} of $X$ along production $p$ and admissible match $m$, and denoted $p_m(X)\xLeftarrow[p,m]{} X$.
\end{defi}

Note that by virtue of Lemma~\ref{lem:convenient}, the object $p_m(X)$ produced via a given derivation of an object $X$ along a linear production $p$ and an admissible match $m$ is \emph{unique up to isomorphism}. From here on, we will refer to linear productions as \emph{linear (rewriting) rules}. Next, we recall the concept of \emph{(concurrent) composition} of linear rules.


\subsection{Concurrent composition and concurrency theorem}%
\label{sec:concur}

Given two linear productions $p_1,p_2\in \Lin{\bfC}$ and an object $X\in \obj{\bfC}$, it is intuitively clear that one may consider acting with $p_2$ on a produced object $X'=p_{1_{m_1}}(X)$ (for some admissible match $m_1$). However, there exists also the interesting possibility to consider first composing the \emph{rules} in a certain sense, and then applying the \emph{sequential composite} to the object $X$. To this end, consider the following well-known definition.

\begin{defi}[DPO-type concurrent composition~\cite{lack2005adhesive}]\label{def:DPOcc}
  Let $p_1,p_2\in \Lin{\bfC}$ be two linear productions. Then a span ${\color{h1color}\mathbf{m}}=(I_2{\color{h1color}\xleftarrow{m_2} M_{21}\xrightarrow{m_1}}O_1)$ with $m_1,m_2\in \cM$---where we use the {\color{h1color}blue colouring} to signify the overlap of $p_1$ and $p_2$---is called\footnote{In the DPO rewriting literature, admissible matches of rules are also referred to as \emph{dependency relations}.} an \emph{admissible match of $p_2$ into $p_1$}, denoted $\mathbf{m}\in \RMatch{p_2}{p_1}$, if in the diagram below the squares marked $\mathsf{POC}$ are constructable as pushout complements (where the cospan $I_2\xrightarrow{n_2}N_{21}\xleftarrow{n_1}O_1$ is obtained by taking the pushout marked ${\color{h1color}\mathsf{PO}}$):
  \begin{equation}\label{eq:DPOccomp}
    \begin{mycd}
      O_2\ar[d,"n_2^{*}"'] &
      K_2 \ar[l,"o_2"']\ar[r,"i_2"]\ar[d,"k_2"']
        \ar[dl,phantom,"\mathsf{PO}"] &
      I_2 \ar[dr,h1color,bend right,"n_2"]\ar[dl,phantom,"\mathsf{POC}"] &
      {\color{h1color}M_{21}}
        \ar[l,h1color,"m_2"']\ar[r,h1color,"m_1"]\ar[d,h1color,phantom,"\mathsf{PO}"] &
      O_1 \ar[dl,h1color,bend left,"n_1"']\ar[dr,phantom,"\mathsf{POC}"]&
      K_1 \ar[l,"o_1"']\ar[r,"i_1"]\ar[d,"k_1"]
        \ar[dr,phantom,"\mathsf{PO}"] &
      I_1\ar[d,"n_1^{*}"]\\
      {\color{h2color}O_{21}} &
      \overline{K}_2\ar[l,"o_2'"]\ar[rr,"i_2'"'] & &
      {\color{h1color}N_{21}} & &
      \overline{K}_1\ar[ll,"o_1'"]\ar[r,"i_1'"] & {\color{h2color}I_{21}}\\
    \end{mycd}
  \end{equation}
  In this case, we write\footnote{It follows from the properties of $\cM$-adhesive categories (i.e.\ stability of $\cM$-morphisms under pushouts) that all morphisms in~\eqref{eq:DPOccomp} are $\cM$-morphisms, whence the span $\comp{p_2}{\mathbf{m}}{p_1}$ is a span of $\cM$-morphisms, and thus indeed an element of $\Lin{\bfC}$.} $\comp{p_2}{\mathbf{m}}{p_1}\in \Lin{\bfC}$ for the \emph{composite} of $p_2$ with $p_1$ along the admissible match $\mathbf{m}$, defined as
  \begin{equation}
    \comp{p_2}{\mathbf{m}}{p_1}\equiv {\color{h2color}(O_{21}\xleftarrow{o_{21}}K_{21}\xrightarrow{i_{21}}I_{21})}:=
    ({\color{h2color}O_{12}}\xleftarrow{o_2'}\overline{K}_2\xrightarrow{i_2'}{\color{h1color}N_{21}})\circ
    ({\color{h1color}N_{21}}\xleftarrow{o_1'}\overline{K}_1\xrightarrow{i_1'}{\color{h2color}I_{21}})\,
  \end{equation}
  where we have used the {\color{h2color}orange colouring} to emphasise the components of the composite production, and where $\circ$ denotes the operation of span composition (cf.\ Lemma~\ref{lem:spans}).
\end{defi}

The following theorem is a refinement of a well-known result from the literature, where the novel feature of our version that will prove quintessential in the following is the specification of the theorem via \emph{admissible matches of linear rules} (rather than the less specific notion of $E$-concurrent derivations as in the work of Ehrig et al.~\cite{EHRIG:2014ma}). In the adhesive category setting, this approach had already been investigated in~\cite{lack2005adhesive}. The reason our modification (which hinges on  Assumption~\ref{as:cats}) provides a strong improvement over the traditional results resides in the fact that in the synthesis step (see below), one is not only led to derive a certain \emph{cospan} encoding the causal interaction of the two sequentially applied rules, but in fact a \emph{span} of $\cM$-morphisms that is unique up to isomorphism, and that thus in a certain sense provides a \emph{minimal} encoding of said causal interaction. Besides practical advantages, this result is in particular strictly necessary in order to lift the notion of associativity of sequential compositions, DPO-type rule algebras and canonical representations of DPO-type rule algebras from the adhesive to the $\cM$-adhesive setting.

\begin{thm}[{DPO-type Concurrency Theorem; modification of~\cite[Thm.~4.17]{EHRIG:2014ma}, compare~\cite[Thm.~7.11]{lack2005adhesive}}]\label{thm:concur}
Let $\bfC$ be an $\cM$-adhesive category satisfying Assumption~\ref{as:cats}. Let $p_1,p_2\in \Lin{\bfC}$ be two linear rules and $X_0\in \obj{\bfC}$ an object.
\begin{itemize}
\item \textbf{Synthesis:} Given a two-step sequence of derivations
\begin{equation*}
X_2\xLeftarrow[p_2,m_2]{} X_1\xLeftarrow[p_1,m_1]{}X_0\,,
\end{equation*}
with $X_1:=p_{1_{m_1}}(X_0)$ and $X_2:=p_{2_{m_2}}(X_1)$, there exists a composite rule $q=\comp{p_2}{\mathbf{n}}{p_1}$
for a unique $\mathbf{n}\in \RMatch{p_2}{p_1}$,
 and a unique admissible match $n\in \Match{q}{X}$, such that
 \begin{equation*}
  q_n(X)\xLeftarrow[q,n]{} X_0\qquad \text{and}\qquad q_n(X_0)\cong X_2\,.
 \end{equation*}
\item \textbf{Analysis:} Given an admissible match $\mathbf{n}\in \RMatch{p_2}{p_1}$ of $p_2$ into $p_1$ and an admissible match $n\in \Match{q}{X}$ of the composite $q=\comp{p_2}{\mathbf{n}}{p_1}$ into $X$, there exists a unique pair of admissible matches $m_1\in \Match{p_1}{X_0}$ and $m_2\in \Match{p_2}{X_1}$ (with $X_1:=p_{1_{m_1}}(X_0)$) such that
\begin{equation*}
    X_2\xLeftarrow[p_2,m_2]{} X_1 \xLeftarrow[p_1,m_1]{} X_0\qquad \text{and}\qquad
    X_2\cong q_n(X)\,.
\end{equation*}
\end{itemize}
\begin{proof}
  --- \textbf{Synthesis:} %
  Consider the setting presented in~\eqref{eq:CTs1}. %
  Here, we have obtained the candidate match ${\color{h1color}\mathbf{n}}=(I_2{\color{h1color}\leftarrow M_{21}\rightarrow} O_1)$ via pulling back the cospan $(I_2{\color{h1color}\rightarrow} X_1{\color{h1color}\leftarrow} O_1)$. %
  Next, we construct ${\color{h1color}N_{21}}$ via taking the pushout of ${\color{h1color}\mathbf{n}}$, which induces a unique arrow ${\color{h1color}N_{21}\rightarrow}X_1$. Crucially, it follows from Assumption~\ref{as:cats} that this arrow is in the class $\cM$. %
  The diagram in~\eqref{eq:CTs2} is obtained by taking the pullbacks of the spans $\overline{K}_i\rightarrow X_1{\color{h1color}\leftarrow N_{21}}$ (obtaining the objects $K_i'$, for $i=1,2$), followed by letting ${\color{h2color}O_{21}}:=\pO{O_2\leftarrow K_2\rightarrow K_2'}$ and ${\color{h2color}I_{21}}:=\pO{O_1\leftarrow K_1\rightarrow K_1'}$. %
  By virtue of pushout-pullback (Lemma~\ref{lem:convenient}) and pushout-pushout decomposition (Lemma~\ref{lem:pushoutpullback}), respectively, the resulting squares are all pushouts. The final step as depicted in~\eqref{eq:CTs3} consists in constructing ${\color{h2color}K_{21}}=\pB{K_2'\rightarrow {\color{h1color}N_{21}}\leftarrow K_1'}$ and ${\color{orange}\overline{K}_{21}}=\pB{\overline{K_2}\rightarrow X_1\leftarrow \overline{K_1}}$, which by universality of pullbacks induces a unique arrow ${\color{h2color}K_{21}\rightarrow \overline{K}_{21}}$. %
  By invoking pullback decomposition (Lemma~\ref{lem:convenient}) and the $\cM$-van Kampen property (cf.\ Def.~\ref{def:adhc}) twice, one may demonstrate that the squares $\cSquare{{\color{h1color}K_{21}},{\color{h1color}\overline{K}_{21}},\overline{K}_i,K_i'}$ (for $i=1,2$) are pushouts. Thus the claim follows by invoking pushout pasting according to Lemma~\ref{lem:pushoutpullback} twice in order to obtain the pushout squares $\cSquare{{\color{h1color}K_{21}},{\color{h1color}\overline{K}_{21}},X_2,{\color{h1color}O_{21}}}$ and  $\cSquare{{\color{h1color}K_{21}},{\color{h1color}\overline{K}_{21}},X_0,{\color{h1color}I_{21}}}$.

  --- \textbf{Analysis:} Given the setting as depicted in~\eqref{eq:CTa1}, we may obtain the configuration of~\eqref{eq:CTa2} by letting $\overline{K}_i=\pO{K_i'\leftarrow{\color{h2color}K_{21}\rightarrow \overline{K}_{21}}}$ (for $i=1,2$). By virtue of pushout decomposition (Lemma~\ref{lem:pushoutpullback}), the resulting new squares are all pushouts. Next, by constructing\footnote{Since the construction is entirely symmetric in this step, we could have equivalently chosen to define $X_1=\pO{\overline{K}_2\leftarrow K_2'\rightarrow {\color{h1color}N_{21}}}$.} $X_1=\pO{\overline{K}_1\leftarrow K_1'\rightarrow {\color{h1color}N_{21}}}$, we obtain the diagram in~\eqref{eq:CTa3}. Since $\cSquare{{\color{h1color}K_{21}},{\color{h1color}\overline{K}_{21}},X_1,{\color{h1color}N_{21}}}$ and $\cSquare{{\color{h1color}K_{21}},{\color{h1color}\overline{K}_{21}},\overline{K}_2,K_2'}$ are pushouts, by pushout decomposition so is $\cSquare{K_2',\overline{K}_2,X_1,{\color{h1color}N_{21}}}$. Thus we finally arrive at the configuration in~\eqref{eq:CTa4} via compositions of pushout squares, thus concluding the proof.
\end{proof}
\end{thm}

The details of the above proof permit to easily derive the following technical corollary:
\begin{cor}\label{lem:adm1}
  Under the assumptions of Theorem~\ref{thm:concur}, every configuration such as in the lower part of the diagram in~\eqref{eq:CTa1}, whence the commutative sub-diagram form by the two pushout squares below,
  \begin{equation}
    \begin{mycd}
      {\color{h2color}O_{21}}
        \ar[d,h2color] &
      {\color{h2color}K_{21}}
        \ar[l,h2color]\ar[r,h2color]\ar[d,h2color]
        \ar[dl,h2color,phantom,"\mathsf{PO}"]
        \ar[dr,h2color,phantom,"\mathsf{PO}"]&
      {\color{h2color}I_{21}}
        \ar[d,h2color] \\
      X_2 &
      {\color{h2color}\overline{K}_{21}}
        \ar[l,h2color]\ar[r,h2color]&
      X_0
    \end{mycd}\,,
  \end{equation}
   uniquely induce the configuration of four adjacent pushout squares presented the lower back part of~\eqref{eq:CTa3}, whence
   \begin{equation}
    \begin{mycd}
      {\color{h2color}O_{21}}
        \ar[d,h2color] &
        K_2' \ar[l]\ar[r,h1color]\ar[d]
        \ar[dl,phantom,"\mathsf{PO}"]
        \ar[dr,phantom,"\mathsf{PO}"]&
      {\color{h1color}N_{21}} \ar[d,h1color]
        &
        K_1' \ar[r]\ar[l,h1color]\ar[d]\ar[dl,phantom,"\mathsf{PO}"]
        \ar[dr,phantom,"\mathsf{PO}"]&
      {\color{h2color}I_{21}}
        \ar[d,h2color] \\
      X_2 &
      \overline{K}_2 \ar[l]\ar[r] &
      X_1 &
      \overline{K}_1 \ar[l]\ar[r] &
      X_0
    \end{mycd}\,,
   \end{equation}
   and vice versa.
\end{cor}


\begin{figure}[ht!]
\begin{subequations}
\begin{align}
  \vcenter{\hbox{\includegraphics[scale=0.6,page=1]{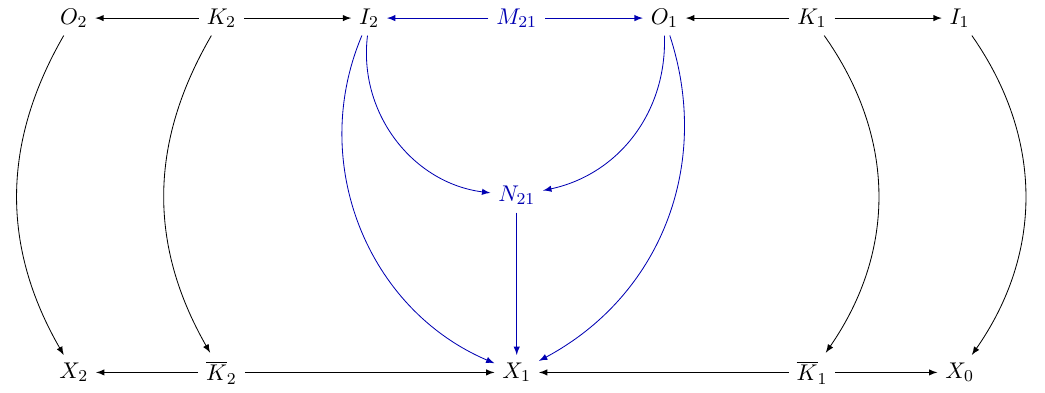}}}\label{eq:CTs1}\\
  \vcenter{\hbox{\includegraphics[scale=0.6,page=2]{images/concurrency-proof.pdf}}}\label{eq:CTs2}\\
  \vcenter{\hbox{\includegraphics[scale=0.6,page=3]{images/concurrency-proof.pdf}}}\label{eq:CTs3}
\end{align}
\end{subequations}
\caption{\label{fig:CTs} \emph{Synthesis} part of the concurrency theorem.}
\end{figure}

\begin{figure}[ht!]
\begin{subequations}
\begin{align}
  \vcenter{\hbox{\includegraphics[scale=0.6,page=1]{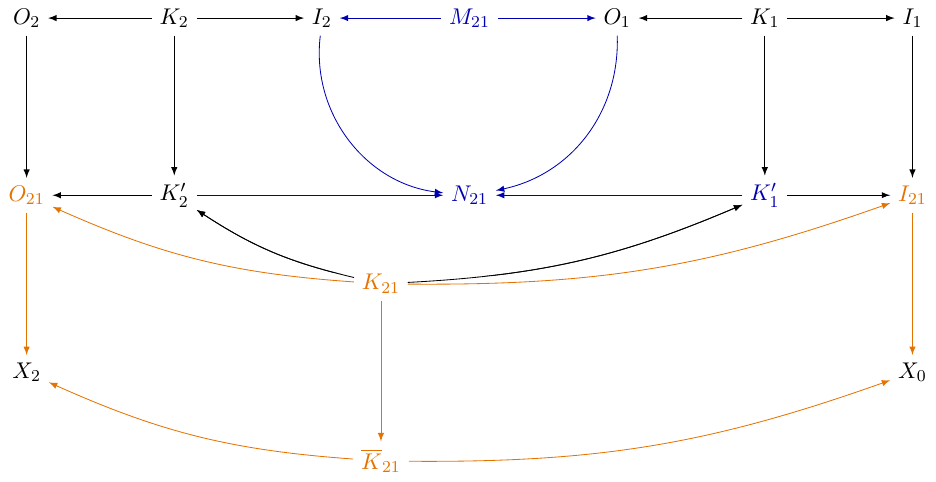}}}\label{eq:CTa1}\\
  \vcenter{\hbox{\includegraphics[scale=0.6,page=2]{images/concurrency-proof-Analysis.pdf}}}\label{eq:CTa2}\\
  \vcenter{\hbox{\includegraphics[scale=0.6,page=3]{images/concurrency-proof-Analysis.pdf}}}\label{eq:CTa3}\\
  \vcenter{\hbox{\includegraphics[scale=0.6,page=4]{images/concurrency-proof-Analysis.pdf}}}\label{eq:CTa4}
\end{align}
\end{subequations}
\caption{\label{fig:CTa} \emph{Analysis} part of the concurrency theorem.}
\end{figure}


\subsection{Concurrent composition and associativity}%
\label{sec:assoc}


While the concurrency theorem (Theorem~\ref{thm:concur}) for DPO rewriting is classical, to the best of our knowledge the following result is new. It states a certain form of associativity for compositions of linear productions.

\begin{thm}[DPO-type associativity theorem]\label{thm:assocDPO}
  Let $\bfC$ be an $\cM$-adhesive category that satisfies Assumption~\ref{as:cats}. %
  Then the composition operation $\comp{.}{.}{.}$ on linear productions of $\bfC$ is \emph{associative} in the following sense: %
given linear productions $p_1,p_2,p_3\in \Lin{\bfC}$, there exists a bijective correspondence
between pairs of admissible matches $(\mathbf{m}_{21},\mathbf{m}_{3(21)})$ and $(\mathbf{m}_{32},\mathbf{m}_{(32)1})$ such that
\begin{equation}\label{eq:THMassoc}
  \comp{p_3}{\mathbf{m}_{3(21)}}{\left(\comp{p_2}{\mathbf{m}_{21}}{p_1}\right)}\; \cong \;
  \comp{\left(\comp{p_3}{\mathbf{m}_{32}}{p_2}\right)}{\mathbf{m}_{(32)1}}{p_1}\,.
\end{equation}
\begin{proof}
Since DPO derivations are symmetric, it suffices to show one side of the correspondence. Our proof is constructive, demonstrating how, given a pair of admissible matches
\begin{equation}
\begin{aligned}
  {\color{h1color}\mathbf{m}_{21}}
    &=(O_2{\color{h1color}\leftarrow M_{21}\rightarrow }I_1)\in \RMatch{p_2}{p_1}\\
  {\color{h1color}\mathbf{m}_{3(21)}}
    &=(O_3{\color{h1color}\leftarrow M_{3(21)}\rightarrow }I_{21})\in \RMatch{p_3}{p_{21}}\,,\qquad p_{21}=\comp{p_2}{{\color{h1color}\mathbf{m}_{21}}}{p_1}\,,
\end{aligned}
\end{equation}
one may uniquely (up to isomorphisms) construct from this information a pair of admissible matches
\begin{equation}
\begin{aligned}
  {\color{h1color}\mathbf{m}_{32}}
    &=(O_3{\color{h1color}\leftarrow M_{32}\rightarrow }I_2)\in \RMatch{p_3}{p_2}\\
  {\color{h1color}\mathbf{m}_{(32)1}}
    &=(O_{32}{\color{h1color}\leftarrow M_{(32)1}\rightarrow }I_{1})\in \RMatch{p_{32}}{p_1}\,,\qquad p_{32}=\comp{p_3}{{\color{h1color}\mathbf{m}_{32}}}{p_2}\,,
\end{aligned}
\end{equation}
and such that the property described in~\eqref{eq:THMassoc} holds. We begin by forming the composite rule $p_{3(21)}=\comp{p_3}{\mathbf{m}_{3(21)}}{p_{21}}$, which results in the diagram
\begin{equation}
\vcenter{\hbox{\includegraphics[scale=0.475,page=1]{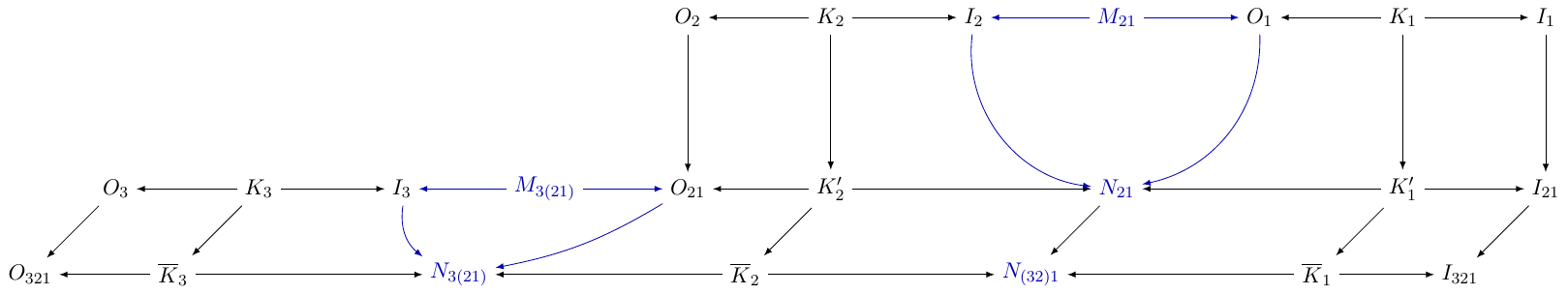}}}
\end{equation}
if we invoke Corollary~\ref{lem:adm1} to construct the four rightmost squares on the bottom. Constructing the pullback ${\color{h1color}M_{32}}=\pB{{\color{h1color}M_{3(21)}\rightarrow}O_{21}\leftarrow O_2}$ (which by universality of pullbacks also leads to an arrow ${\color{h1color}M_{32}\rightarrow}I_3$) and forming the three additional vertical squares on the far left in the evident fashion in the diagram below
\begin{equation}
\vcenter{\hbox{\includegraphics[scale=0.475,page=2]{images/assocProof.pdf}}}
\end{equation}
allows us to construct ${\color{h1color}N_{32}}=\pO{I_3{\color{h1color}\leftarrow M_{32}\rightarrow}O_2}$, which in turn via universality of pushouts uniquely induces an arrow ${\color{h1color}N_{32}\rightarrow N_{3(21)}}$:
\begin{equation}
\vcenter{\hbox{\includegraphics[scale=0.475,page=3]{images/assocProof.pdf}}}
\end{equation}
Here, the rightmost three squares on the top are formed in the evident fashion (and are pushouts according to Lemma~\ref{lem:auxPOPBcat}), while the other arrows of the above diagram are constructed as follows:
\begin{equation}
\begin{aligned}
  K_3'&=\pB{\overline{K}_3\rightarrow{\color{h1color}N_{3(21)}\leftarrow N_{32}}}\,,\quad & &
  O_{32}&=\pO{K_3'\leftarrow K_3\rightarrow O_3}\\
  K_2''&=\pB{{\color{h1color}N_{32}\rightarrow N_{3(21)}\leftarrow}\overline{K}_2}\,,\quad & &
  I_{32}&=\pO{K_2''\leftarrow K_2\rightarrow I_2}
\end{aligned}
\end{equation}
Invoking pushout-pullback and pushout-pushout decomposition repeatedly, it may be verified that all squares thus created on the top and in the front are pushout squares. Defining the pullback object ${\color{h1color}M_{(32)1}}=\pB{I_{32}{\color{h1color}\rightarrow N_{3(21)}\leftarrow}O_1}$, thus inducing an arrow ${\color{h1color}M_{21}\rightarrow M_{3(21)}}$,
\begin{equation}\label{eq:AssocLMstep}
\vcenter{\hbox{\includegraphics[scale=0.475,page=4]{images/assocProof.pdf}}}
\end{equation}
it remains to verify that the square $\cSquare{{\color{h1color}M_{3(21)}},I_{32},{\color{h1color}N_{3(21)}},O_1}$ is not only a pullback, but also a pushout square. To this end, construct the auxiliary diagram depicted in Figure~\ref{fig:DPOrevAuxDiag}, where the top, back, bottom and front cubes that are formed via the newly added arrows compared to~\eqref{eq:AssocLMstep} are also drawn separately for clarity, with suitable 3d rotations applied such as to facilitate the application of further steps in the proof based upon the $\cM$-van Kampen property\footnote{On a philosophical note, it might be worth observing that while sequential compositions of rules are essentially described by two-dimensional commutative diagrams, this final step of the associativity proof appears to have an inherently \emph{three-dimensional} character, in that the properties of the commutative cubes in question delicately rely on each other as described in the proof.}. Note in particular that the four additional new arrows exist due to universality of pullbacks.

\afterpage{%
  \begin{landscape}
    \begin{figure}%
    \centering
\includegraphics[scale=0.7,page=5]{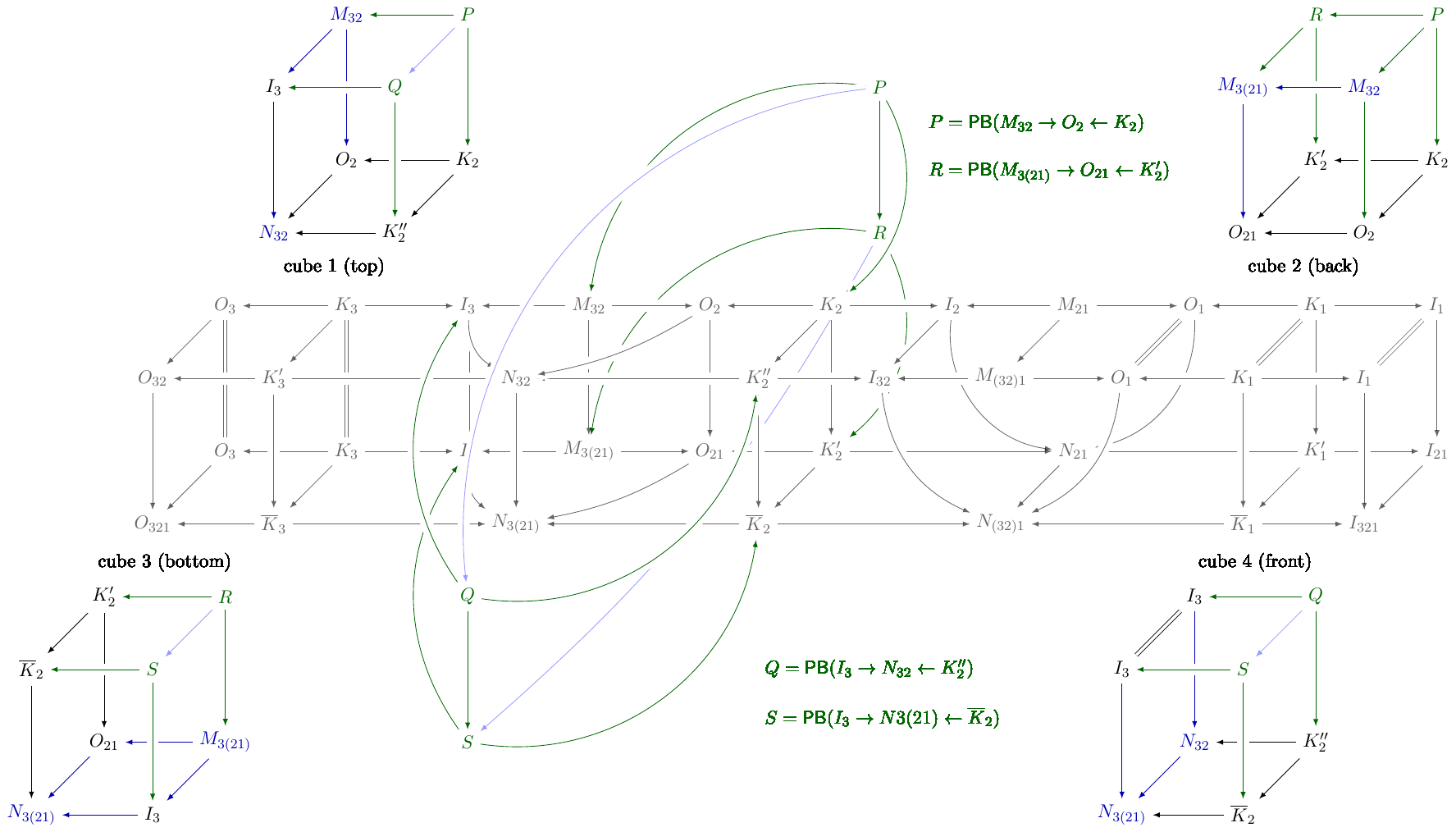}
\caption{\label{fig:DPOrevAuxDiag}Auxiliary diagram for the second part of the DPO associativity proof.}
    \end{figure}
  \end{landscape}%
}
Invoking pullback decomposition as well as the $\cM$-van Kampen property repeatedly, the new commutative cube on the top (i.e.\ the one sitting over the two pushout squares $\cSquare{M_{32},I_3,N_{32},O_2}$ and $\cSquare{K_2,O_2,N_{32},K_2''}$) and the new commutative cube on the bottom (i.e.\ the one sitting under the two pushout squares $\cSquare{M_{3(21)},I_3,N_{3(21)},O_{(21)}}$ and $\cSquare{K_2',O_{21},N_{3(21)},\overline{K}_2}$) have pushouts on all of their faces.

As for the new cubes in the front and back, note first that by virtue of Lemma~\ref{lem:auxPOPBcat} the back left square $\cSquare{I_3,I_3,N_{3(21)},N_{32}}$ of the front cube is a pullback, while the square $\cSquare{M_{32},M_{3(21)},O_{21},O_2}$ had been constructed as a pullback in the main part of the proof. Thus invoking pullback decomposition twice, we may conclude that also the squares $\cSquare{Q,S,\overline{K}_2,K_2''}$ in the front and $\cSquare{P,R,I_{32},K_2}$ in the back are pullbacks, whence invoking the $\cM$-van Kampen twice allows to conclude that the squares $\cSquare{Q,S,I_3,I_3}$ in the front left and $\cSquare{P,R,M_{3(21)},M_{32}}$ in the back left are pushouts. Moreover, since isomorphisms are stable under pushouts by virtue of Lemma~\ref{lem:auxPOPBcat}, we may conclude that $Q\cong S$. We collect all of this information into the following diagram:
\begin{equation*}
\vcenter{\hbox{\includegraphics[scale=0.75]{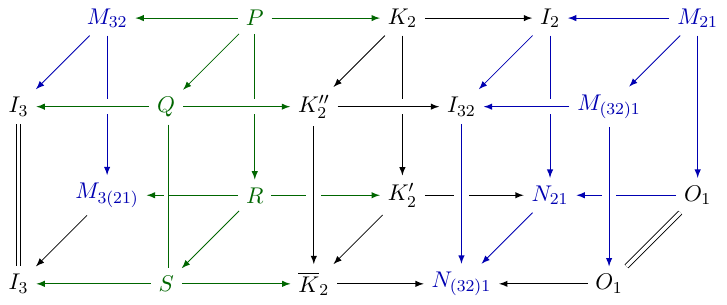}}}
\end{equation*}
To prepare the final steps, let us perform the following ``splitting'' of the above diagram:
\begin{equation*}
\vcenter{\hbox{\includegraphics[scale=0.75]{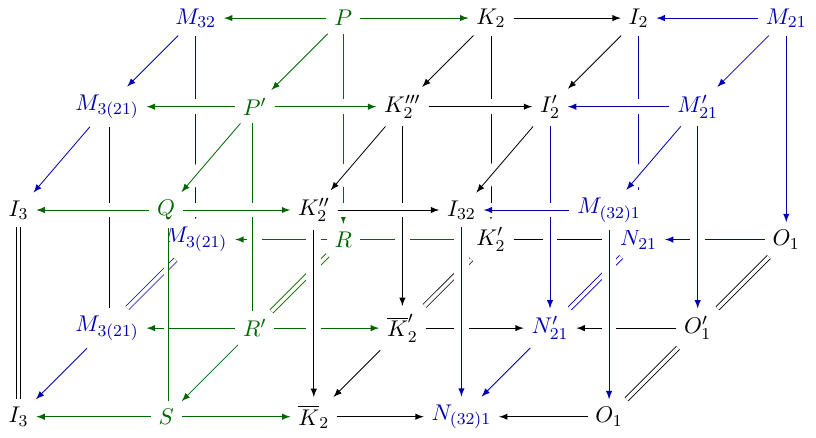}}}
\end{equation*}
We start the construction from the very left: evidently $\cSquare{I_3,M_{3(21)},M_{3(21)},I_3}$ is both a pullback and a pushout. %
Next, construct the pullbacks $P'=\pB{M_{3(21)}\rightarrow I_3\leftarrow Q}$ and $R'=\pB{M_{3(21)}'\rightarrow I_3\leftarrow S}$; by pushout-pullback decomposition, they split the pushout square $\cSquare{P,Q,I_3,M_{23}}$ on the top and $\cSquare{R,S,I_3,M_{3(21)}}$ into two pushout squares each. The latter also implies that $R'\cong R$.

Pasting pushouts, we have that $\cSquare{M_{3(21)}',R',S,I_3}$ is a pushout, whence by pushout-pushout decomposition so is $\cSquare{P',R',S,Q}$ (and thus $P'\cong R'$).

Next, construct the two pushouts $K_3''=\pO{P'\leftarrow P\rightarrow K_2}$ in the top and $\overline{K}_2'=\pO{R'\leftarrow R\rightarrow K_2'}$ on the bottom (which implies $\overline{K}_2'\cong K_2'$). Pushout-pushout decomposition then entails that $\cSquare{P',K_2''',K_2'',Q}$ and $\cSquare{R',\overline{K}_2',\overline{K}_2,S}$ are pushouts, and consequently so is the square $\cSquare{K_2''',\overline{K}_2',\overline{K}_2,K_2''}$.

We repeat the construction of the previous step and construct the pushouts $I_2'=\pO{K_2'''\leftarrow K_2\rightarrow I_2}$ and $N_{21}'=\pO{\overline{K}_2'\leftarrow K_2'\rightarrow N_{21}}$ (which implies that $N_{21}'\cong N_{21}$). Pushout-pushout decomposition then yields the pushout squares $\cSquare{K_2''',I_2',I_{32},K_2''}$ and $\cSquare{\overline{K}_2',N_{21}',N_{(32)1},\overline{K}_2}$, and thus also $\cSquare{I_2',N_{21}',N_{(32)1},I_{32}}$ is a pushout.

Next, split the pullback squares $\cSquare{M_{21},M_{(32)1},I_{32},I_2}$ in the top and \[\cSquare{O_1,O_1,N_{(32)1},N_{21}}\] on the bottom into two pullback squares each via pullback-pullback decomposition, whence via letting $M_{21}'=\pB{I_2'\rightarrow I_{32}\leftarrow M_{(32)1}}$ and $O_1'=\pB{N_{21}'\rightarrow N_{3(21)}\leftarrow O_1}$. By virtue of Lemma~\ref{lem:auxPOPBcat} (i.e.\ stability of isomorphisms under pullbacks), this entails that $O_1'\cong O_1$, whence the square $\cSquare{O_1,O_1',N_{21}',N_{21}}$ is a pushout.

By pullback-pullback decomposition, the square $\cSquare{M_{21}',O_1',N_{21}',I_2'}$ is a pullback. By pushout-pullback decomposition, since by virtue of the previous step the square \[\cSquare{M_{21},O_1',N_{21}',I_2}\] is a pushout and $\cSquare{M_{21}',O_1',N_{21}',I_2'}$ a pullback, $\cSquare{M_{21}',O_1',N_{21}',I_2'}$ is also a pushout.

Finally, since by pushout pasting the square $\cSquare{M_{21}',O_1',N_{(32)1},I_{32}}$ is a pushout, and since $\cSquare{M_{(32)1},O_1,N_{(32)1},I_{32}}$ is by construction a pullback, pushout-pullback decomposition entails that $\cSquare{M_{(32)1},O_1,N_{(32)1},I_{32}}$ is also a pushout, which concludes the proof.
\end{proof}
\end{thm}

In summary, the associativity property manifests itself in the following form, whereby the data provided along the path highlighted in orange below permits to uniquely compute the data provided along the path highlighted in blue (with both sets of overlaps computing the same ``triple composite'' production):
\begin{equation}
\vcenter{\hbox{\includegraphics[scale=0.475,page=1]{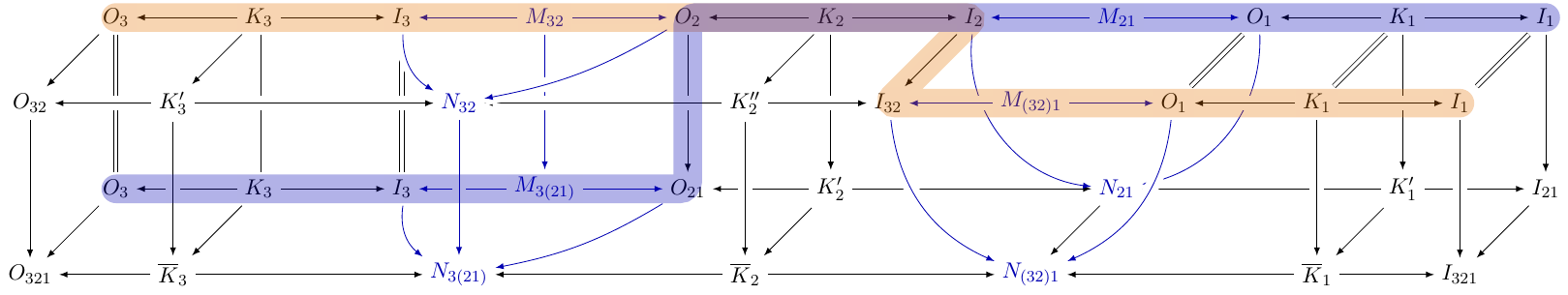}}}
\end{equation}


\section{From associativity of concurrent derivations to rule algebras}%
\label{sec:ACDrd}


In DPO rewriting, each linear rewriting rule has a non-deterministic effect when acting on a given object, in the sense that there generically exist multiple possible choices of admissible match of the rule into the object. One interesting way of incorporating this non-determinism into a mathematical rewriting framework is motivated by the physics literature:
\begin{itemize}[label=$\triangleright$]
\item Each linear rule is lifted to an element of an abstract \emph{vector space}.
\item Concurrent composition of linear rules is lifted to a \emph{bilinear multiplication operation} on this abstract vector space, endowing it with the structure of an \emph{algebra}.
\item The action of rules on objects is implemented by mapping each linear rule (seen as an element of the abstract algebra) to an endomorphism on an abstract vector space whose basis vectors are in bijection with the objects of the adhesive category.
\end{itemize}
While this recipe might seem somewhat ad hoc, we will demonstrate in Section~\ref{sec:HW} that it recovers in fact one of the key constructions of quantum physics and enumerative combinatorics, namely we recover the well-known Heisenberg-Weyl algebra and its canonical representation.

\medskip
Let us first fix the precise type of categories for which our constructions are well-defined. A very general class of such categories is covered by the following set of assumptions:

\begin{asm}\label{ass:RAdpo}
  We assume from hereon that $\bfC$ is a \emph{finitary $\cM$-adhesive category} with $\cM$-effective unions.
\end{asm}

Let us next quote some concepts of general mathematics of key relevance to the material presented in this section:
\begin{defi}[cf.\ e.g.\ \cite{hazewinkel2004algebras}]
  Let $V\equiv(V,+,\cdot)$ be a \emph{$\bK$-vector space}, with $\bK$ a field (e.g.\ $\bK=\bR$). Then a \emph{$\bK$-algebra} $(V,*)$ is defined via equipping the $\bK$-vector space $V$ with a \emph{bilinear binary operation} $*:V\times V\rightarrow V$. Here, bilinearity entails that
  \begin{equation}
  \forall a,b\in V\,,\;\alpha,\beta\in \bK:\quad (\alpha\cdot a)*(\beta\cdot b)=(\alpha\beta)\cdot(a*b)\,.
  \end{equation}
  The algebra $(V,*)$ is called \emph{associative} if
  \begin{equation}
    \forall a,b,c\in V:\quad a*(b*c)=(a*b)*c\,.
  \end{equation}
  It is called \emph{unital} if there exists an element $u\in V$ (then referred to as a \emph{unit element}) such that
  \begin{equation}
    \forall v\in V:\quad u*v=v*u=v\,.
  \end{equation}
  Let $W\equiv(W,+,\cdot)$ be another $\bK$-vector space (over the same field $\bK$ as the $\bK$-vector space $V$), and denote by $End_{\bK}(W)$ the \emph{algebra of endomorphisms over $W$}. Then a \emph{representation} \[
  \rho: (V,*)\rightarrow End_{\bK}(W)
  \]
  of the associative $\bK$-algebra $(V,*)$ is defined as an \emph{algebra homomorphism}, such that
  \begin{equation}
    \forall a,b\in V:\quad \rho(a*b)=\rho(a)\rho(b)\,.
  \end{equation}
  If $(V,*)$ is in addition unital, we also require that $\rho$ maps the unit element $u\in V$ to the identity element of $End_{\bK}(W)$,
  \begin{equation}
    \rho(u)=\mathbb{1}_{End_{\bK}(W)}\,.
  \end{equation}
\end{defi}

\begin{defi}
Let $\cR_{\bfC}\equiv(\cR_{\bfC},+,\cdot)$ be the free $\bR$-vector space on $\Lin{\bfC}$, defined concretely via a bijection $\delta:\Lin{\bfC}_{\cong}\rightarrow \cR_{\bfC}$ from \emph{isomorphism classes of linear productions} in $\Lin{\bfC}$ to the set of basis vectors of $\cR_{\bfC}$. In order to distinguish between elements of $\Lin{\bfC}$ and of $\cR_{\bfC}$, we introduce the notation
\begin{equation}
(\grule{O}{r}{I}):=\delta\left(O\xleftharpoonup{r}I\right)\,.
\end{equation}
We will later refer to $\cR_{\bfC}$ as the $\bR$-vector space of \emph{rule algebra elements}.
\end{defi}

\begin{defi}\label{def:RADPO}
Define the \emph{DPO rule algebra product} $*_{\cR_{\bfC}}$ on an $\cM$-adhesive category $\bfC$ satisfying Assumption~\ref{ass:RAdpo} as the binary operation
\begin{equation}
*_{\cR_{\bfC}}:\cR_{\bfC}\times \cR_{\bfC}\rightarrow \cR_{\bfC}:(R_1,R_2)\mapsto R_1*_{\cR_{\bfC}} R_2\,,
\end{equation}
where for two basis vectors $R_i=\delta(p_i)$ encoding the linear rules $p_i\in \Lin{\bfC}$ ($i=1,2$),
\begin{equation}\label{eq:defRcomp}
R_2*_{\cR_{\bfC}}R_1
:=\sum_{\mathbf{m}\in \RMatch{p_2}{p_1}}\delta\left(\comp{p_2}{\mathbf{m}}{p_1}\right)\,.
\end{equation}
Here, we take the notational convention that $\sum_{\emptyset}\dotsc=0_{\cR_{\bfC}}$ (i.e.\ the summation over an empty set of admissible matches evaluates to the zero element of the vector space $\cR_{\bfC}$). The definition is extended to arbitrary (finite) linear combinations of basis vectors by bilinearity, whence for $p_i,p_j\in \Lin{\bfC}$ and $\alpha_i,\beta_j\in \bR$,
\begin{equation}
\left(\sum_i \alpha_i\cdot\delta(p_i)\right)*_{\cR_{\bfC}}\left(\sum_j\beta_j\cdot \delta(p_j)\right):=\sum_{i,j}(\alpha_i\cdot\beta_j)\cdot \left(\delta(p_i)*_{\cR_{\bfC}}\delta(p_j)\right)\,.
\end{equation}
We refer to $\cR_{\bfC}\equiv(\cR_{\bfC},*_{\cR_{\bfC}})$ as the \textbf{\emph{rule algebra}} (of linear DPO-type rewriting rules over the $\cM$-adhesive category $\bfC$).
\end{defi}

It is worthwhile noting that if the category $\bfC$ possesses an $\cM$-initial object, the ``trivial match'' of two linear productions $p_j=(O_j\leftarrow K_j\rightarrow I_j)$ (for $j=1,2$), i.e.\ $\mathbf{m}_{\mIO}=(I_2\leftarrow\mIO\rightarrow O_1)$, may be verified to be always an admissible match according to the definition of the DPO-type concurrent composition of productions (Definition~\ref{def:DPOcc}) and by virtue of Lemma~\ref{lem:binaryCoproducts}.

\begin{thm}
For every category $\bfC$ satisfying Assumption~\ref{ass:RAdpo}, the associated DPO-type rule algebra $\cR_{\bfC}\equiv(\cR_{\bfC},*_{\cR_{\bfC}})$ is an \emph{associative algebra}. If $\bfC$ in addition possesses an \emph{$\cM$-initial object} $\mIO\in \obj{\bfC}$, $\cR_{\bfC}$ is in addition a \emph{unital algebra}, with unit element $R_{\mIO}:=(\grule{\mIO}{id_{\mIO}}{\mIO})$.
\end{thm}
\begin{proof}
Associativity follows immediately from the associativity of the operation $\comp{.}{.}{.}$ proved in Theorem~\ref{thm:assocDPO}. The claim that $R_{\mIO}$ is the unit element of the rule algebra $\cR_{\bfC}$ of an adhesive category $\bfC$ with strict initial object follows directly from the definition of the rule algebra product for $R_{\mIO}*_{\cR_{\bfC}}R$ and $R*_{\cR_{\bfC}}R_{\mIO}$ for $R\in \cR_{\bfC}$. For clarity, we present below the category-theoretic composition calculation that underlies the equation $R_{\mIO}*_{\cR_{\bfC}}R=R$:
\begin{align}\gdef\mycdScale{0.85}\tabularnewline
    &\begin{mycd}
      \mIO\ar[d] &
      \mIO \ar[l,equal]\ar[r,equal]\ar[d]
        \ar[dl,phantom,"\mathsf{PO}"] &
      \mIO \ar[dr,h1color,bend right]
      \ar[dl,phantom,"\mathsf{POC}"] &
      {\color{h1color}\mIO}
        \ar[l,h1color,equal]
        \ar[r,h1color]\ar[d,h1color,phantom,"\mathsf{PO}"] &
      O \ar[dl,h1color,bend left,equal]\ar[dr,phantom,"\mathsf{POC}"]&
      K \ar[l,"o"']\ar[r,"i"]\ar[d,equal]
        \ar[dr,phantom,"\mathsf{PO}"] &
      I\ar[d,equal]\\
      {\color{h2color}O} &
      O\ar[l,equal]\ar[rr,equal] & &
      {\color{h1color}O} & &
      K\ar[ll,"o"]\ar[r,"i"'] & {\color{h2color}I}\\
    \end{mycd}\\
  &({\color{h2color}O}\xleftarrow{id_O}O\xrightarrow{id_O}{\color{h1color}O})\circ ({\color{h1color}O}\xleftarrow{o}K\xrightarrow{i}{\color{h2color}I})
  ={\color{h2color} (O\xleftarrow{o}K\xrightarrow{i}I)}
  \qedhere
  \end{align}
\end{proof}
The property of a rule algebra being unital and associative has the important consequence that one can provide  \emph{representations} for it.
The following definition, given at the level of $\cM$-adhesive categories with $\cM$-initial objects, captures several of the concrete notions
of canonical representations in the physics literature; in particular, it generalises the concept of canonical representation of the Heisenberg-Weyl algebra as explained in Section~\ref{sec:HW}. Intuitively, since a given linear rewriting rule $r\in \Lin{\bfC}$ may in general be applied in multiple different ways to a given object $X\in \obj{\bfC}$, i.e.\ via the different choices of admissible matches $m\in \Match{r}{X}$, one might wish to encode this non-determinism in some form or other. The variant of encoding chosen via the rule algebra and canonical representation (see below) consists heuristically in constructing a \emph{``sum over all outcomes $r_m(X)$ of applying rule $r$ to object $X$ via admissible matches $m$''}. Since a particular DPO derivation from $X$ to $r_m(X)$ according to Definition~\ref{def:DPOr} defines the object $r_m(X)$ only up to universal isomorphism (of the relevant pushout complement and pushout in~\eqref{eq:DPOr}), it is a priori clear that we must make precise the concept as a form of \emph{``sum over all outcomes $r_m(X)$ up to isomorphism''}. At this point, one might in principle envision multiple possible design choices, yet the construction of canonical representation introduced below will have certain practical advantages:
\begin{itemize}[label=$\triangleright$]
  \item Computing the sum over all outcomes possible via first applying rule $r_1$ via all possible admissible matches to an object $X$, followed by summing over all possible ways to apply rule $r_2$ to the outcomes of the first step, one may alternatively via the Concurrency Theorem (Theorem~\ref{thm:concur}) first pre-compute the sum over all ways to compose $r_2$ with $r_1$ (computed concretely via the rule algebra product operation $\delta(r_2)*_{\cR_{\bfC}}\delta(r_1)$), followed by taking the sum over applying the rules encoded in $\delta(r_2)*_{\cR_{\bfC}}\delta(r_1)$ to $X$ along all admissible matches. The canonical representation $\rho_{\cR_{\bfC}}$ defined below is precisely the ``canonical choice'' to make this concept concrete.
  \item One of the key motivations for this particular method of construction was to recover the important example of the Heisenberg-Weyl algebra as a special case of the general DPO algebra and canonical representation construction (see Section~\ref{sec:HW} and further details therein).
\end{itemize}

\begin{defi}\label{def:canRep}
Let $\bfC$ be a category satisfying Assumption~\ref{ass:RAdpo}, and which in addition possesses an $\cM$-initial object $\mIO\in \obj{\bfC}$, and let $\cR_{\bfC}$ be its associated rule algebra of DPO type. Denote by $\hat{\bfC}\equiv(\hat{\bfC},+,\cdot)$ the \emph{$\bR$-vector space of objects of $\bfC$}, defined via a bijection $\ket{.}:\obj{\bfC}_{\cong}\rightarrow \hat{\bfC}$ from isomorphism classes of objects of $\bfC$ to the set of basis vectors of $\hat{C}$. Then the \emph{canonical representation} $\rho_{\bfC}$ of $\cR_{\bfC}$ is defined as the algebra homomorphism
$\rho_{\bfC}:\cR_{\bfC}\rightarrow End(\hat{\bfC})$, with
\begin{equation}\label{eq:canRep}
\rho_{\bfC}(p)\ket{C}:=\begin{cases}
\sum_{m\in \Match{p}{C}}\ket{p_m(C)}\quad &\text{if }\Match{p}{C}\neq \emptyset\\
0_{\hat{\bfC}}&\text{otherwise,}
\end{cases}
\end{equation}
extended to generic elements of $\cR_{\bfC}$ and of $\hat{\bfC}$ by linearity.
\end{defi}
The fact that $\rho_C$ as given in Definition~\ref{def:canRep} is an algebra homomorphism is shown below.
\begin{thm}[Canonical Representation]\label{thm:canRep}
For $\bfC$ a category satisfying Assumption~\ref{ass:RAdpo} and with $\cM$-initial object, $\rho_{\bfC}: \cR_{\bfC} \rightarrow End(\hat{\bfC})$ of Definition~\ref{def:canRep} is a homomorphism of unital associative algebras.
\end{thm}
\begin{proof}
In order for $\rho_{\bfC}$ to qualify as an algebra homomorphism (of unital associative algebras $\cR_{\bfC}$ and $End(\hat{\bfC})$), we must have (with $R_{\mIO}=\delta(r_{\mIO})$, $r_{\mIO}=\GRule{\mIO}{id_{\mIO}}{\mIO}$)
\begin{equation*}
(i)\; \rho_{\bfC}(R_{\mIO})=\mathbb{1}_{End(\hat{\bfC})}\quad \text{and}\quad (ii)\;\forall R_1,R_2\in \cR_{\bfC}:\; \rho_{\bfC}(R_1*_{\cR_{\bfC}}R_2)=\rho_{\bfC}(R_1)\rho_{\bfC}(R_1)\,.
\end{equation*}
Due to linearity, it suffices to prove the two properties on basis elements $\delta(p),\delta(q)$ of $\cR_{\bfC}$ and on basis elements $\ket{C}$ of $\hat{\bfC}$. Property $(i)$ follows directly from the definition,
\begin{equation*}
\forall C\in \obj{\bfC}_{\cong}:\quad \rho_{\bfC}(R_{\mIO})\ket{C}\overset{\eqref{eq:canRep}}{=}\sum_{m\in\Match{r_{\mIO}}{C}}\ket{{(r_{\mIO})}_m(C)}=\ket{C}\,.
\end{equation*}
Property $(ii)$ follows from Theorem~\ref{thm:concur} (the Concurrency Theorem): for all basis elements $\delta(p),\delta(q)\in \cR_{\bfC}$ (with $p,q\in \Lin{\bfC}$) and for all $C\in \obj{\bfC}_{\cong}$,
\begin{align*}
\rho_{\bfC}\left(\delta(q)*_{\bfC}\delta(p)\right)\ket{C}
&\overset{\eqref{eq:defRcomp}}{=}
  \sum_{\mathbf{d}\in \RMatch{q}{p}}
  \rho_{\bfC}\left(\delta\left(\comp{q}{\mathbf{d}}{p}\right)\right)\ket{C}\\
&\overset{\eqref{eq:canRep}}{=}
  \sum_{\mathbf{d}\in \RMatch{q}{p}}\;
  \sum_{e\in\Match{r_{\mathbf{d}}}{C}}\ket{{(r_{\mathbf{d}})}_e(C)}\quad \tag{$r_{\mathbf{d}}=\comp{q}{\mathbf{d}}{p}$}\\
&=
\sum_{m\in \cM_{p}(C)}\sum_{n\in\Match{q}{p_{m}(C)}}
\ket{q_n(p_m(C))} \tag{via Thm.~\ref{thm:concur}}\\ 
&\overset{\eqref{eq:canRep}}{=}
\sum_{m\in \Match{p}{C}}\rho_{\bfC}\left(\delta(q)\right)\ket{p_m(C)}\\
&\overset{\eqref{eq:canRep}}{=}
\rho_{\bfC}\left(\delta(q)\right)\rho_{\bfC}\left(\delta(p)\right)\ket{C}\,.
\tag*{\qedhere}
\end{align*}%
\end{proof}


\subsection{Recovering the blueprint: the Heisenberg-Weyl algebra}%
\label{sec:HW}


As a first consistency check and interesting special (and arguably simplest) case of rule algebras, consider the adhesive category $\mathbb{F}$ of equivalence classes of finite sets, and functions. This category might alternatively be interpreted as the category $\mathbf{G}_0$ of isomorphism classes of \emph{finite discrete graphs}, whose linear rules are precisely the injective partial morphisms of discrete graphs. Specialising to a subclass or linear rules, namely to those with a trivial context object,
\begin{equation*}
\GRule{O}{\emptyset}{I}\equiv (O\leftarrow \emptyset\rightarrow I)\,,
\end{equation*}
we recover the famous Heisenberg-Weyl algebra and its canonical representation.

\begin{defi}[c.f.\ e.g.\ \cite{blasiak2011combinatorial}]
The associative unital \emph{Heisenberg-Weyl algebra} over $\bR$ is defined as
\begin{equation}
  A_{HW}:=\frac{\bR[\pi^{\dag},\pi]}{\langle [\pi,\pi^{\dag}]-id\rangle}\,,
\end{equation}
where $\bR[\pi^{\dag},\pi]$ denotes the polynomial ring over $
\bR$ with two generators $\pi^{\dag}$ and $\pi$ (the ``creator'' and the ``annihilator'') quotiented by the ideal generated by the \emph{canonical commutation relation}
\begin{equation}
[\pi,\pi^{\dag}]\equiv \pi\pi^{\dag}-\pi^{\dag}\pi=id\,.
\end{equation}
Two well-known constructions of representations of $A_{HW}$ illustrate the important role of this algebra in applications of combinatorics and physics:
\begin{itemize}[label=$\triangleright$]
  \item The \emph{Bargmann-Fock (BF) representation} is defined as the algebra homomorphism from $A_{HW}$ to the space of endomorphisms over $\bR[z]$ (the vector space of polynomials in the formal variable $z$) that maps $\pi^{\dag}$ to the linear operator $\hat{z}$, and $\pi$ to $\tfrac{\partial}{\partial z}$, with
  \begin{equation}
   \forall n\in \bZ_{\geq0}:\quad  \hat{z} z^n:=z^{n+1}\,,\quad
   \tfrac{\partial}{\partial z}z^0:=0\,,\quad \tfrac{\partial}{\partial z}z^n:=nz^{n-1}\; (n>0)\,.
  \end{equation}
  The canonical commutation relation may then be explicitly verified to hold on each basis vector $z^n$:
  \begin{equation}
  \begin{aligned}
    [\tfrac{\partial}{\partial z},\hat{z}]z^0
    &=\tfrac{\partial}{\partial z}\hat{z}z^0-\hat{z}\tfrac{\partial}{\partial z}z^0=\tfrac{\partial}{\partial z}z^1-0=z^0\\
    [\tfrac{\partial}{\partial z},\hat{z}]z^n
    &=\tfrac{\partial}{\partial z}\hat{z}z^n-\hat{z}\tfrac{\partial}{\partial z}z^0=\tfrac{\partial}{\partial z}z^{n+1}-n\hat{z}z^{n-1}=z^n\; (n>0)
    \end{aligned}
  \end{equation}
  \item The so-called \emph{canonical representation} of the HW algebra (see (IV) in Theorem~\ref{thm:HW}), which is manifestly isomorphic to the BF representation via $\hat{z}\leftrightarrow a^{\dag}$, $\tfrac{\partial}{\partial z}\leftrightarrow a$ and $z^n\leftrightarrow \ket{n}$, is the formulation typically encountered in the combinatorics and physics literature (the latter especially in the context of quantum mechanics).
\end{itemize}
\end{defi}

\noindent
We will now provide a realisation of the Heisenberg-Weyl algebra directly within the DPO rule algebra formalism. Reserving the symbol $\cH$ for this realisation of the algebra to avoid confusion, note in particular that we give a very intuitive and \emph{intrinsic} meaning to the ``creator'' and the ``annihilator'' (as simply the rule algebra elements associated to the linear rules of creating and of deleting a vertex, respectively).

\begin{defi}\label{def:HW}
Let $\cR_0\equiv \cR_{\mathbf{G}_0}$ denote the rule algebra of DPO type rewriting for finite discrete graphs. Then the subalgebra $\cH$ of $\cR_0$ is defined as the algebra whose elementary generators are
\begin{equation}
x^{\dag}:=(\grule{\bullet}{\emptyset}{\emptyset})\equiv\delta(\bullet\leftarrow \emptyset\rightarrow \emptyset)\,,\quad
x:=(\grule{\emptyset}{\emptyset}{\bullet})\equiv\delta(\emptyset\leftarrow \emptyset\rightarrow \bullet)\,,
\end{equation}
and whose elements are (finite) linear combinations of words in $x^{\dag}$ and $x$ (with concatenation given by the rule algebra multiplication $*_{\cR_0}$) and of the unit element $R_{\emptyset}=(\grule{\emptyset}{\emptyset}{\emptyset})$. The canonical representation of $\cH$ is the restriction of the canonical representation of $\cR_0$ to $\cH$.
\end{defi}
The following theorem demonstrates how well-known properties of the Heisenberg-Weyl algebra (see e.g.\ \cite{blasiak2011combinatorial,bdgh2016,bdp2017} and references therein) follow directly from the previously introduced constructions of the rule algebra and its canonical representation. This justifies our claim that the Heisenberg-Weyl construction is a special case of our general framework.
\begin{thm}[Heisenberg-Weyl algebra from discrete graph rewriting rule algebra]\label{thm:HW}\hfill
\begin{enumerate}[(I)]
\item For integers $m,n>0$,
\begin{equation}
\underbrace{x^{\dag}*_{\cR_0}\dotsc*_{\cR_0}x^{\dag}}_{\text{$m$ times}}=\underbrace{x^{\dag}\uplus \dotsc\uplus x^{\dag}}_{\text{$m$ times}}\,,\quad \underbrace{x*_{\cR_0}\dotsc*_{\cR_0}x}_{\text{$n$ times}}=\underbrace{x\uplus \dotsc\uplus x}_{\text{$n$ times}}\,,
\end{equation}
where we define for linear rules $p_1,p_2\in Lin(\bfC)$
\begin{equation}
\delta(p_1)\uplus\delta(p_2):=\delta(\comp{p_1}{\emptyset}{p_2})\,.
\end{equation}
\item The generators $x,x^{\dag}\in \cH$ fulfil the \emph{canonical commutation relation}
\begin{equation}
[x,x^{\dag}]\equiv x*_{\cR_0}x^{\dag}-x^{\dag}*_{\cR_0}x=R_{\emptyset}\,,\quad R_{\emptyset}=(\grule{\emptyset}{\emptyset}{\emptyset})\,.
\end{equation}
\item Every element of $\cH$ may be expressed as a (finite) linear combination of so-called \emph{normal-ordered} expressions $x^{\dag\:*r}*x^{*s}$ (with $r,s\in \bZ_{\geq0}$ and $*\equiv *_{\cR_0}$). By convention, $x^{\dag\:*r}:=x^{\dag}*\dotsc *x^{\dag}$ (a product of $r$ generators $x^{\dag}$), and analogously for $x^{*\:s}$.
\item Denoting by $\ket{n}\equiv\ket{\bullet^{\uplus\:n}}$ ($n\in \bZ_{\geq 0}$) the basis vector associated to the discrete graph with $n$ vertices in the vector space $\hat{\bfG}_0$ of isomorphism classes of discrete graphs, the canonical representation of $\cH$ according to Definition~\ref{def:canRep} reads explicitly
\begin{equation}
a^{\dag}\ket{n}=\ket{n+1}\,,\quad a\ket{n}=\begin{cases}
n\cdot\ket{n-1}\quad &\text{if } n>0\\
0_{\hat{G}_0}&\text{else}
\end{cases}\,,
\end{equation}
with $a^{\dag}:=\rho_{\mathbf{G}_0}(x^{\dag})$ (the \emph{creation operator}) and $a:=\rho_{\mathbf{G}_0}(x)$ (the \emph{annihilation operator}).
\end{enumerate}
\end{thm}
\newpage
\begin{proof}~\begin{enumerate}[(I)]
\item Since there is no partial injection possible between the input of one copy and the output of another copy of $x^{\dag}$ other than the trivial match, and similarly for two copies of $x$, the claim follows by induction.
\item Computing the commutator $[x,x^{\dag}]=x*x^{\dag}-x^{\dag}*x$ (with $*\equiv*_{\cR_0}$) explicitly, we find that
\begin{equation}
x*x^{\dag}= x\uplus x^{\dag}+R_{\emptyset}\,,\quad
x^{\dag}*x=x^{\dag}\uplus x\,,
\end{equation}
from which the claim follows due to commutativity of the operation $\uplus$ on $\cR_0$, $x\uplus x^{\dag}=x^{\dag}\uplus x$. Here, the contribution $R_{\emptyset}$ arises from the following sequential composition:
\begin{equation}\label{eq:HWproofAux}\gdef\mycdScale{0.85}
\begin{aligned}
&\begin{mycd}
      \emptyset\ar[d] &
      \emptyset \ar[l]\ar[r]\ar[d]
        \ar[dl,phantom,"\mathsf{PO}"] &
      \OneVertG[] \ar[dr,h1color,bend right]
      \ar[dl,phantom,"\mathsf{POC}"] &
      {\OneVertG[h1color]}
        \ar[l,h1color]
        \ar[r,h1color]\ar[d,h1color,phantom,"\mathsf{PO}"] &
      \OneVertG[] \ar[dl,h1color,bend left]
      \ar[dr,phantom,"\mathsf{POC}"]&
      \emptyset \ar[l]\ar[r]\ar[d]
        \ar[dr,phantom,"\mathsf{PO}"] &
      \emptyset\ar[d]\\
      {\color{h2color}\emptyset} &
      \emptyset\ar[l]\ar[rr] & &
      {\OneVertG[h1color]} & &
      \emptyset\ar[ll]\ar[r] & {\color{h2color}\emptyset}\\
    \end{mycd}\\
  &({\color{h2color}\emptyset}\leftarrow\emptyset\rightarrow \OneVertG[h1color])\circ (\OneVertG[h1color]\leftarrow \emptyset \rightarrow{\color{h2color}\emptyset})
  ={\color{h2color} (\emptyset\leftarrow\emptyset\rightarrow\emptyset)}
  \end{aligned}
  \end{equation}
\item It suffices to prove the statement for basis elements of $\cH$. Consider thus an arbitrary composition of a finite number of copies of the generators $x$ and $x^{\dag}$. Then by repeated application of the commutation relation $[x,x^{\dag}]=R_{\emptyset}$, and since $R_{\emptyset}$ is the unit element for $*_{\cR_0}$, we can convert the arbitrary basis element of $\cH$ into a linear combination of normal-ordered elements. More explicitly, from the viewpoint of the composition operation on linear rules, an expression of the form $x^{\dag\:*r}*x^{*\:s}$ is easily verified to evaluate to
\begin{equation}
x^{\dag\:*r}*x^{*\:s}=\delta\left(\bullet^{\uplus\:r}\leftarrow \emptyset\rightarrow \bullet^{\uplus\:s}\right)\,.
\end{equation}
In full analogy to the computation presented in~\eqref{eq:HWproofAux}, composing $x^{\dag\:*r}*x^{*\:s}$ with $x^{\dag\:*k}*x^{*\:\ell}$ will evaluate to
\begin{equation}
x^{\dag\:*r}*x^{*\:s}*x^{\dag\:*k}*x^{*\:\ell}=\sum_{n=0}^{\min(s,k)}\tfrac{s! k!}{(s-n)!n!(k-n)!}\, x^{\dag\:*(r+k-n)}*x^{*\:(s+\ell-n)}\,.
\end{equation}
Note in particular that the coefficient of a term $x^{\dag\:*(r+k-n)}*x^{*\:(s+\ell-n)}$ in the above sum coincides with the number of ways to match $n$ of the vertices between a discrete graph with $s$ vertices (i.e.\ the input interface of the rule encoded in $x^{\dag\:*r}*x^{*\:s}$) and a discrete graph with $k$ vertices (i.e.\ the output interface of the rule encoded in $x^{\dag\:*k}*x^{*\:\ell}$). Thus the DPO algebra implementation precisely explains the full combinatorics involved in the problem of ``normal-ordering'' expressions in the HW algebra.
\item Note first that by definition $\ket{0}=\ket{\emptyset}$. To prove the claim that for all $n\geq 0$
\begin{equation*}
a^{\dag}\ket{n}=\ket{n+1}\,,
\end{equation*}
we apply Definitions~\ref{def:DPOr} and~\ref{def:canRep} by computing the following diagram (compare~\eqref{eq:DPOr}): there exists precisely one admissible match of the empty graph $\emptyset\in \obj{\mathbf{G}_0}$ into the $n$-vertex discrete graph $\OneVertG[]^{\uplus\:n}$, whence constructing the pushout complement marked with dashed arrows and the pushout marked with dotted arrows we verify the claim: 
\begin{equation*}\gdef\mycdScale{0.85}
\begin{mycd}
\OneVertG[] \ar[d,dotted]
\ar[dr,phantom,"{\mathsf{PO}}"] &
\emptyset \ar[l]\ar[r]\ar[d,dotted]
\ar[dr,phantom,"{\mathsf{POC}}"] &
\emptyset \ar[d,"\exists!"]\\
\OneVertG[]^{\uplus\:(n+1)} & 
\OneVertG[]^{\uplus\:n}\ar[l,dashed]\ar[r,dashed] & 
\OneVertG[]^{\uplus\:n}\\ 
\end{mycd}
\end{equation*}
Proceeding analogously in order to prove the formula for the representation $a=\rho_{\mathbf{G_0}}(x)$,
\begin{equation*}
  a\ket{n}:=\begin{cases}
  n\cdot\ket{n-1}\quad &\text{if } n>0\\
  0_{\hat{G}_0} &\text{else,}
  \end{cases}
\end{equation*}
we find that for $n>0$ there exist $n$ admissible matches of the $1$-vertex graph $\OneVertG[]$ into the $n$-vertex graph $\OneVertG[]^{\uplus\:n}$, for each of which the application of the rule $\GRule{\OneVertG[]}{}{\emptyset}$ along the match results in the graph $\OneVertG[]^{\uplus\:(n-1)}$: 
\begin{equation*}\gdef\mycdScale{0.85}
\begin{mycd}
\emptyset
\ar[d,dotted]\ar[dr,phantom,"{\mathsf{PO}}"] &
\emptyset \ar[l]\ar[r]\ar[d,dotted]\ar[dr,phantom,"{\mathsf{POC}}"] &
\OneVertG[] \ar[d,"\text{$n$ different matches}"]\\
\OneVertG[]^{\uplus\:(n-1)} & 
\OneVertG[]^{\uplus\:(n-1)}\ar[l,dashed]\ar[r,dashed] & 
\OneVertG[]^{\uplus\:n}\\ 
\end{mycd}
\; \Rightarrow\;
\forall n>0:a\ket{\OneVertG[]^{\uplus\: n}}=n\cdot \ket{\OneVertG[]^{\uplus\:(n-1)}} 
\end{equation*}
Finally, for $n=0$, since by definition there exists no admissible match from the $1$-vertex graph $\OneVertG$ into the empty graph $\emptyset$, whence indeed
\begin{equation*}
a\ket{\emptyset}=\rho_{\mathbf{G_0}}\left(\grule{\emptyset}{\emptyset}{\OneVertG[]}\right)\ket{\emptyset}=0_{\hat{\mathbf{G}}_0}\,.
\qedhere
\end{equation*}
\end{enumerate}
\end{proof}


\subsection{Applications of rule algebras to combinatorics}%
\label{sec:RT}

Here we consider an example application, working with undirected multigraphs.
\begin{defi}[Compare e.g.\ \cite{padberg2017towards}]
  Let $Id_{\mathbf{Set}}:\mathbf{Set}\rightarrow\mathbf{Set}$ be the identity functor on $\mathbf{Set}$, and let $\cP^{(1,2)}:\mathbf{Set}\rightarrow\mathbf{Set}$ be the restricted covariant power set functor (which maps a set $S$ to its subsets of cardinality $1$ or $2$). Then the \emph{category of finite undirected multigraphs} $\mathbf{uGraph}$ is defined as the finitary restriction of a comma category~\cite{ehrig:2006aa},
  \begin{equation}
    \mathbf{uGraph}:={(Id_{\mathbf{Set}},\cP^{(1,2)})}_{{\rm fin}}\,.
  \end{equation}
  An object $U$ of $\mathbf{uGraph}$ is specified via the data $U\equiv(V_U,E_U,inc_U)$, where $V_U$ is a set of vertices, $E_U$ a set of edges, and $inc_U:E_U\rightarrow \cP^{(1,2)}(V_U)$ is the incidence function.
\end{defi}

\begin{lem}
  $\mathbf{uGraph}$ is a weak adhesive HLR category, for $\cM_{\mathbf{uGraph}}$ the class of pairs of monomorphisms in $\mathbf{Set}_{{\rm fin}}$. It has an $\cM_{\mathbf{uGraph}}$-initial object (the empty graph $\emptyset\in \obj{\mathbf{uGraph}}$) as well as $\cM_{\mathbf{uGraph}}$-effective unions.
\end{lem}
\begin{proof}
  The identity functor $Id_{\mathbf{Set}}$ trivially preserves pushouts in $\mathbf{Set}$, while $\cP^{(1,2)}$ preserves pullbacks along monomorphisms in $\mathbf{Set}$ (cf.\ e.g.\ Appendix~A, Cor.~7 of~\cite{padberg2017towards}). Therefore, by~\cite[Thm.~4.15]{ehrig:2006aa}, $(Id_{\mathbf{Set}},\cP^{(1,2)})$ is a weak adhesive HLR category, and then by Theorem~\ref{thm:finRes}, so is its finitary restriction $\mathbf{uGraph}$. $\cM_{\mathbf{uGraph}}$-initiality is trivial, while the property of $\cM_{\mathbf{uGraph}}$-effective unions follows by application of Theorem~\ref{thm:euAux}.
\end{proof}

In order to illustrate the intimate interplay of rule algebraic and combinatorial structures, we will now provide an example where the integer coefficients arising in applications of rules of undirected multigraphs yield a known combinatorial integer sequence.
\begin{defi}
We define the algebra $\cA$ as the one generated\footnote{As in the case of the Heisenberg-Weyl algebra, by ``generated'' we understand that a generic element of $\cA$ is a finite linear combination of (finite) words in the generators and of the identity element $R_{\emptyset}$, with concatenation given by the rule algebra composition.} by the rule algebra elements
\begin{equation}\label{eq:Adef}
e_{+}:=\tfrac{1}{2}\cdot\left(\tP{%
  \vI{1}{1}{black}{}{black}{}
  \vI{1}{2}{black}{}{black}{}
  \eC{1}{1}{=}{black}{}}\right)\,,\quad
e_{-}:=\tfrac{1}{2}\cdot\left(\tP{%
  \vI{1}{1}{black}{}{black}{}
  \vI{1}{2}{black}{}{black}{}
  \eA{2}{1}{=}{black}{}}\right)\,, \quad
d:=\tfrac{1}{2}\cdot\left(\tP{%
  \vI{1}{1}{black}{}{black}{}
  \vI{1}{2}{black}{}{black}{}}\right)\quad (e_{+},e_{-},d\in \cR_{\mathbf{uGraph}})\,.
\end{equation}
For convenience, we adopt here a graphical notation (so-called ``rule diagrams''~\cite{bdg2016}) in which we depict a rule algebra basis element $(\grule{O}{f}{I})\in \cR_{\mathbf{uGraph}}$ as the graph of its induced injective partial morphism $(I\xrightharpoonup{f}O)\in Inj(I,O)$ of graphs $I$ and $O$, with the input graph $I$ drawn at the \emph{bottom}, the output graph $O$ at the \emph{top}, and where the structure of the morphism $f$ is indicated with \emph{dotted lines}. In the example above, $e_{+}$ encodes (up to the factor of $\tfrac{1}{2}$ chosen purely for convenience) a linear rule with input interface given by two vertices; in applying the rule, the vertices are to be kept (indicated by the vertical dotted lines), while a new edge between them is created (indicated by the $\times$ and dotted line to the edge, symbolising ``creation''). Dually, $e_{-}$ encodes the deletion of an edge, while $d$ encodes an identity rule on the pattern of two disjoint vertices. In all three rules, the factor $\tfrac{1}{2}$ is chosen such as to compensate for the symmetry of the three linear rules evident from the depictions (i.e.\ along the horizontal mirror axis of the ``rule diagrams'').
\end{defi}

The algebra thus defined may be characterised\footnote{``Characterised'' here refers to the observation that by utilising the commutation relations given in~\eqref{eq:UAcomm}, it is possible to express an arbitrary element of the algebra as a linear combination in ``normal-ordered terms'' $e_{+}^{*\:p}*e_{-}^{*\:m}*d^{*\:n}$ (for $p,m,n\in \bZ_{\geq0}$).} via its \emph{commutation relations}, which read (with $[x,y]:=x*y-y*x$ for $*\equiv*_{\cR_{\mathbf{uGraph}}}$)
\begin{equation}\label{eq:UAcomm}
[e_{-},e_{+}]=d\,,\quad [e_{+},d]=[e_{-},d]=0\,.
\end{equation}
Here, the only nontrivial contribution (i.e.\ the one that renders the first commutator non-zero) may be computed from the DPO-type composition diagram\footnote{Note that the number indices are used solely to specify the precise structure of the match, and are not to be understood as actual vertex labels or types.}
below and its variant for the admissible match $\TwoVertEdgeGLb{}{1}{2}\leftarrow
\TwoVertEdgeGLb{}{12'}{21'}\rightarrow  \TwoVertEdgeGLb{}{1'}{2'} $:
\begin{equation}\gdef\mycdScale{0.85}
\begin{aligned}
&\begin{mycd}
      \TwoVertG[]\ar[d] &
      \TwoVertG[]
        \ar[l]\ar[r]\ar[d]
        \ar[dl,phantom,"\mathsf{PO}"] &
      \TwoVertEdgeGLb{}{1}{2}
        \ar[dr,h1color,bend right]
        \ar[dl,phantom,"\mathsf{POC}"] &
      {\TwoVertEdgeGLb{h1color}{11'}{22'}}
        \ar[l,h1color]
        \ar[r,h1color]
        \ar[d,h1color,phantom,"\mathsf{PO}"] &
      \TwoVertEdgeGLb{}{1'}{2'}
        \ar[dl,h1color,bend left]
        \ar[dr,phantom,"\mathsf{POC}"]&
      \TwoVertG[]
         \ar[l]\ar[r]\ar[d]
        \ar[dr,phantom,"\mathsf{PO}"] &
      \TwoVertG[]\ar[d]\\
      \TwoVertG[h2color] &
      \TwoVertG[]\ar[l]\ar[rr] & &
      \TwoVertEdgeG[h1color] & &
      \TwoVertG[]\ar[ll]\ar[r] &
      \TwoVertG[h2color]\\
    \end{mycd}\\
  &(\TwoVertG[h2color]\leftarrow\TwoVertG[]\rightarrow \TwoVertEdgeG[h1color])\circ (\TwoVertEdgeG[h1color]\leftarrow \TwoVertG[] \rightarrow\TwoVertG[h2color])
  ={\color{h2color} (\TwoVertG[h2color]\leftarrow\TwoVertG[h2color]\rightarrow\TwoVertG[h2color])}
  \end{aligned}
\end{equation}

We find an interesting structure for the representation of $\cA$:
\begin{lem}
Let $E_{\pm}:=\rho(e_{\pm})$ and $D:=\rho(d)$ (for $\rho\equiv\rho_{\mathbf{uGraph}}$), and let $\hat{\bfG}:=\widehat{\mathbf{uGraph}}$. Denote for each non-negative integer $n\in \bZ_{\geq0}$ by $\hat{\bfG}_n\subset \hat{\bfG}$ the linear subspace of $\hat{\bfG}$ spanned by basis vectors indexed by isomorphism classes of undirected graphs with $n$ vertices. Then the linear endomorphisms $\rho(X)$ for $X\in\{e_{+},e_{-},d\}$ possess the vector spaces $\hat{\bfG}_n$ as \emph{invariant subspaces}, resulting in the \emph{decompositions}
\begin{equation}
\rho(X)=\bigoplus_{n\geq 0} \left(\rho(X)\vert_{\hat{\bfG}_n}\right)\,.
\end{equation}
\begin{proof}
The three rules that define the algebra $\cA$ do not modify the number of vertices when applied to a given graph (via the canonical representation). In other words, for each $X\in \{e_{+},e_{-},d\}$ and for a basis vector $\ket{G}$ of the invariant subspace $\hat{\mathbf{G}}_n$ of graphs with $n$ vertices, the image of $\ket{G}$ under application of $\rho(X)$ is a linear combination of basis vectors again in $\hat{\bfG}_n$, i.e.\ $\rho(X)\ket{G}\subset \hat{\mathbf{G}}_n$.
\end{proof}
\end{lem}

\begin{rem}While at first a rather technical observation, the decomposition of the linear operators $E_{\pm}=\rho(e_{\pm})$ and $D=\rho(d)$ via restriction to their invariant subspaces gives rise to rich combinatorial structures, even though the operators $E_{\pm}$ and $D$ originate from representations of very simple rule algebra elements. Since a subspace $\hat{\bfG}_n$ is characterised by isomorphism classes of finite undirected multigraphs with a fixed number $n$ of vertices, but with an arbitrary (finite) number of edges, each space $\hat{\bfG}_n$ is countably infinite-dimensional. We will exemplify the combinatorial structures that arise for the cases $n=2$ and $n=3$ in the following.
\end{rem}

One may easily verify that the operator $D$ may be equivalently expressed as
\begin{equation}\label{eq:defOv}
D=\tfrac{1}{2}\cdot\rho\left(\tP{%
  \vI{1}{1}{black}{}{black}{}
  \vI{1}{2}{black}{}{black}{}}\right)
  =\tfrac{1}{2}\left(O_{\bullet}O_{\bullet}-O_{\bullet}\right)\,,\quad
  O_{\bullet}:=\rho\left(\tP{\vI{1}{1}{black}{}{black}{}}\right)\,.
\end{equation}
Since the diagonal operator $O_{\bullet}$ when applied to an arbitrary graph state $\ket{G}$ for $G\in \bfG$ effectively counts the number $n_V(G)$ of vertices of $G$,
\begin{equation}
O_{\bullet}\ket{G}=n_V(G)\ket{G}\,,
\end{equation}
one finds that
\begin{equation}
D\ket{G}=\tfrac{1}{2}O_{\bullet}(O_{\bullet}-1)\ket{G}
=\tfrac{1}{2}n_V(G)(n_V(G)-1)\ket{G}\,.
\end{equation}
One may thus alternatively analyse the canonical representation of $\cA$ split into invariant subspaces of $D$. The lowest non-trivial such subspace is the space $\hat{\bfG}_2$ of undirected multigraphs on two vertices. It in fact contains a representation of the Heisenberg-Weyl algebra, with $E_{+}$ and $E_{-}$ taking the roles of the creation and of the annihilation operator, respectively, and with the number vectors $\ket{n}\equiv \ket{\bullet^{\uplus\:n}}$ implemented as follows (with ${(m)}_n:=\Theta(m-n)m!/(m-n)!$): 
\begin{equation}
\begin{aligned}
E^n_{+}\ket{\tP{\node[vertices] (a) at (1,1) {};\node[vertices] (a) at (1.5,1) {};}}
&=\ket{\tP{%
  \node[vertices] (a) at (1,1) {};
  \node[vertices] (b) at (3,1) {};
  \draw (a) edge[bend left=50] (b);
  \draw (a) edge[bend left=40] (b);
  \draw (a) edge[bend left=30] (b);
  \draw (b) edge[bend left=50] (a);
  \node at ($(a.center)!0.5!(b.center)$) {$\vcenter{\hbox{\vdots}} \text{\tiny{$n$ times}}$};}}\,,\quad 
  E_{-}\ket{\tP{%
  \node[vertices] (a) at (1,1) {};
  \node[vertices] (b) at (3,1) {};
  \draw (a) edge[bend left=50] (b);
  \draw (a) edge[bend left=40] (b);
  \draw (a) edge[bend left=30] (b);
  \draw (b) edge[bend left=50] (a);
  \node at ($(a.center)!0.5!(b.center)$) {$\vcenter{\hbox{\vdots}} \text{\tiny{$n$ times}}$};}}= 
  {(n)}_1\ket{\tP{%
  \node[vertices] (a) at (1,1) {};
  \node[vertices] (b) at (3,1) {};
  \draw (a) edge[bend left=50] (b);
  \draw (a) edge[bend left=40] (b);
  \draw (b) edge[bend left=50] (a);
  \node at ($(a.center)!0.5!(b.center)$) {$\vcenter{\hbox{\vdots}} \text{\tiny{$(n-1)$ times}}$};}}\,. 
\end{aligned}
\end{equation}
But already the invariant subspace based on the initial vector $\ket{\tP{%
\node[vertices] (a) at (1,1) {};
\node[vertices] (b) at (1.5,1) {};
\node[vertices] (c) at (2,1) {};}}\in \hat{\bfG}_3$ has a very interesting combinatorial structure:
\begin{equation}
\begin{aligned}
E_{+}\ket{\tP{%
\node[vertices] (a) at (1,1) {};
\node[vertices] (b) at (1.5,1) {};
\node[vertices] (c) at (2,1) {};}}
&=
3\ket{\tP{%
\node[vertices] (a) at (1,1) {};
\node[vertices] (b) at (1.5,1) {};
\node[vertices] (c) at (2,1) {};
\draw (a) edge (b);}}\equiv 3\ket{\{1,0,0\}}\\
E_{+}^2\ket{\tP{%
\node[vertices] (a) at (1,1) {};
\node[vertices] (b) at (1.5,1) {};
\node[vertices] (c) at (2,1) {};}}
&=
3\left(\ket{\tP{%
\node[vertices] (a) at (1,1) {};
\node[vertices] (b) at (1.5,1) {};
\node[vertices] (c) at (2,1) {};
\draw (a) edge[bend left] (b);
\draw (b) edge[bend left] (a);}}
+2\ket{\tP{%
\node[vertices] (a) at (1,1) {};
\node[vertices] (b) at (1.5,1) {};
\node[vertices] (c) at (2,1) {};
\draw (a) edge (b);
\draw (b) edge (c);}}
\right)
\equiv3\left( \ket{\{2,0,0\}}+2\ket{\{1,1,0\}}\right)\\
E_{+}^3\ket{\tP{%
\node[vertices] (a) at (1,1) {};
\node[vertices] (b) at (1.5,1) {};
\node[vertices] (c) at (2,1) {};}}
&=
3\left(\ket{\tP{%
\node[vertices] (a) at (1,1) {};
\node[vertices] (b) at (1.5,1) {};
\node[vertices] (c) at (2,1) {};
\draw (a) edge[bend left] (b);
\draw (a) edge (b);
\draw (b) edge[bend left] (a);}}
+6\ket{\tP{%
\node[vertices] (a) at (1,1) {};
\node[vertices] (b) at (1.5,1) {};
\node[vertices] (c) at (2,1) {};
\draw (a) edge[bend left] (b);
\draw (a) edge[bend right] (b);
\draw (b) edge (c);}}
+2\ket{\tP{%
\node[vertices] (a) at (1,1) {};
\node[vertices] (b) at (1.5,1) {};
\node[vertices] (c) at (2,1) {};
\draw (a) edge[bend left] (c);
\draw (a) edge (b);
\draw (b) edge (c);}}
\right)\\
&\equiv3\left(\ket{\{3,0,0\}}+6\ket{\{2,1,0\}}+2\ket{\{1,1,1\}}\right)\\
&\;\vdots\\
E^n_{+}\ket{\tP{%
\node[vertices] (a) at (1,1) {};
\node[vertices] (b) at (1.5,1) {};
\node[vertices] (c) at (2,1) {};}}&\equiv E^n_{+}\ket{\{0,0,0\}}=3\sum_{k=0}^n T(n,k)\ket{S(n,k)}\,
\end{aligned}
\end{equation}
Here, the state $\ket{\{f,g,h\}}$ with $f\geq g\geq h\geq 0$ and $f+g+h=n$ is the graph state on three vertices with (in one of the possible presentations of the isomorphism class) $f$ edges between the first two, $g$ edges between the second two and $h$ edges between the third and the first vertex. Furthermore,  $T(n,k)$ and $S(n,k)$ are given by the entry \href{https://oeis.org/A286030}{\emph{A286030}} of the OEIS database~\cite{OEISts}. The interpretation of $S(n,k)$ and $T(n,k)$ is that each triple $S(n,k)$ encodes the outcome of a game of three players, counting (without regarding the order of players) the number of wins per player for a total of $n$ games. Here, the second index $k$ in a given $S(n,k)=\{f',g',h'\}$ denotes the position of this triple of integers in the reverse lexicographic order over all triples of integers $\{f,g,h\}$ satisfying the constraints $f\geq g\geq h\geq 0$ and $f+g+h=n$ (see~\cite{OEISts} for further details). Then $T(n,k)/3^{(n-1)}$ gives the probability that a particular pattern $S(n,k)$ occurs in a random sample.

\medskip
While of course the example presented must be seen as just a first proof of concept, many integer sequences as well as certain types of orthogonal polynomials possess an interpretation as being related to the counting of certain graphical structures (cf.\ e.g.\ \cite{strehl2017lacunary} and references therein). It thus appears to be an interesting avenue of future research to investigate the apparently quite intricate interrelations between rule algebra representation theory and combinatorics, which suggests in particular to reinterpret the graphical methods of importance in enumerative combinatorics via a direct encoding within the rule algebra framework.


\subsection{Applications of rule algebras to stochastic mechanics}%
\label{sec:SM}


One of the main motivations that underpinned the development of the rule algebra framework prior to this paper~\cite{bdg2016,bdgh2016} has been the link between associative unital algebras of transitions and continuous-time Markov chains (CTMCs). Famous examples of such particular types of CTMCs include chemical reaction systems (see e.g.\ \cite{bdp2017} for a recent review) and stochastic graph rewriting systems (see~\cite{bdg2016} for a rule-algebraic implementation). With our novel formulation of unital associative rule algebras and their canonical representation for $\cM$-adhesive categories, it is possible to specify a \emph{general stochastic mechanics framework}. While a detailed presentation of the far-reaching consequences of this result is relegated to~\cite{bdg2018}, it suffices here to define the basic framework and to indicate the potential of the idea with a short worked example.

\begin{rem}
  For the readers not familiar with the theory of \emph{continuous-time Markov chains (CTMCs)}, the salient points of the mathematical construction are as follows:
  \begin{itemize}[label=$\triangleright$]
    \item Fixing a suitable $\cM$-adhesive category $\bfC$, the CTMCs to be defined will evolve over a space of \emph{(sub-)probability distributions} indexed by isomorphism classes of objects of $\bfC$ (see~\eqref{eq:defProbC}). 
    \item In general, these distributions may (and in many examples will) have \emph{countably infinite support}. Therefore, one must utilise certain concepts introduced in the mathematical theory of CTMCs (notably \emph{(sub-)stochastic operators} and \emph{Fr\`{e}chet spaces}) in order to explain the passage from linear operators (acting on elements of vector spaces, i.e.\ on finite linear combinations of basis vectors) to the countably infinitely supported setting. 
    \item Specifically, a type of operator called \emph{infinitesimal generator} must be constructed that will play a crucial role as the operator governing the evolution of the CTMC\@. In the setting at hand, a central result of CTMC theory entails that this type of operator is fully characterised by its ``matrix element structure'' (see~\eqref{eq:defHctmc}), i.e.\ its off-diagonal elements must be non-negative, while the diagonal elements must evaluate to minus the sum of the off-diagonal elements in a given row (and this sum must be finite).
    \item Finally, a key feature of CTMC theory for stochastic rewriting systems is its degree of freedom in choosing a set of \emph{observed properties} (typically a choice of patterns to count) of the system at hand in its time-evolution. Unlike in the special setting of CTMCs arising from stochastic rewriting systems on \emph{discrete graphs} with vertices of possibly different types (where naturally the counts of the numbers vertices of the different types is the ``canonical'' choice of observed property), in generic rewriting systems it is not clear from the outset which properties are interesting to observe. We will first introduce the mathematical notion of \emph{observables} as certain diagonal linear operators in Definition~\ref{def:obs}, followed by an illustration in a concrete worked example at the end of this section.
  \end{itemize}
\end{rem}

\noindent
We begin our construction by specialising the general definition of continuous-time Markov chains (see e.g.\ \cite{norris}) to the setting of rewriting systems (compare~\cite{bdg2016,bdp2017}).

\begin{defi}
Consider an $\cM$-adhesive category $\bfC$ with $\cM$-initial object $\mIO\in \obj{\bfC}$ and satisfying Assumption~\ref{ass:RAdpo}, and let $\hat{\bfC}$ denote the free $\bR$-vector space indexed by isomorphism classes of objects of $\bfC$ according to Definition~\ref{def:canRep}. Then we define the space $\Prob{\bfC}$ as the \emph{space of sub-probability distributions} in the following sense:
\begin{equation}\label{eq:defProbC}
\Prob{\bfC}:=\left.\left\{
\ket{\Psi}=\sum_{o\in \obj{\bfC}_{\cong}}\psi_o \ket{o}
\right\vert
\forall o\in \obj{\bfC}_{\cong}: \psi_o\in \bR_{\geq0}
\land
\sum_{o\in \obj{\bfC}_{\cong}}\psi_o\leq 1
\right\}\,.
\end{equation}
Let $\Stoch{\bfC}:=End(\Prob{\bfC})$ be the space of \emph{sub-stochastic operators}, and denote by $\cS_{\bfC}$ the space\footnote{The space $\cS_{\bfC}$ is referred to as the \emph{Fr\`{e}chet space} of real-valued sequences $f\equiv{(f_o)}_{o\in \obj{\bfC}_{\cong}}$ with semi-norms $\|f\|_{o}:=|f_o|$. Strictly speaking, we are thus tacitly assuming here that the category $\bfC$ is \emph{essentially small}, i.e.\ that it possesses a countable set of isomorphism classes.} of real-valued sequences indexed by isomorphism classes of objects of $\bfC$, and where all coefficients are \emph{finite}. Then a \textbf{\emph{continuous-time Markov chain (CTMC)}} is specified in terms of a tuple of data $(\ket{\Psi(0)},H)$, where
\begin{itemize}[label=$\triangleright$]
  \item $\ket{\Psi(0)}\in \Prob{\bfC}$ is the \emph{initial state}, and where
  \item $H\in End_{\bR}(\cS_{\bfC})$ is the \emph{infinitesimal generator} or \emph{Hamiltonian} of the CTMC\@.
\end{itemize}
The linear operator $H$ is required to be an infinitesimal (sub-)stochastic operator, whence to fulfil the following constraints on its ``matrix elements'' $h_{o,o'}$: 
\begin{equation}\label{eq:defHctmc}
\begin{aligned}
H&\equiv {(h_{o,o'})}_{o,o'\in \obj{\bfC}_{\cong}}\quad \forall o,o'\in \obj{\bfC}_{\cong}:\\
&\quad (i)\; h_{o,o}\leq 0\,,\quad
 (ii) \forall o\neq o':\; h_{o,o'}\geq 0\,,\quad
(iii)\; \sum_{o'} h_{o,o'}=0\,.
\end{aligned}
\end{equation}
According to the mathematical theory of CTMCs~\cite{norris}, under the above conditions this data encodes the \emph{evolution semi-group} $\cE:\bR_{\geq 0}\rightarrow \Stoch{\bfC}$ as the (point-wise minimal non-negative) solution of the \emph{Kolmogorov backwards} or \emph{master equation}:
\begin{equation}
\begin{aligned}
\tfrac{d}{dt}\cE(t)&=H\cE(t)\,,\quad \cE(0)=\mathbb{1}_{End_{\bR}(\cS_{\bfC})}
\Rightarrow\quad &\forall t,t'\in \bR_{\geq 0}: \cE(t)\cE(t')=\cE(t+t')\,.
\end{aligned}
\end{equation}
Consequently, the \emph{time-dependent state} $\ket{\Psi(t)}$ of the system is given by
\begin{equation}
\forall t\in \bR_{\geq 0}:\quad \ket{\Psi(t)}=\cE(t)\ket{\Psi(0)}\,.
\end{equation}
\end{defi}

Typically, our interest in analysing a given CTMC will consist in studying the dynamical statistical behaviour of so-called \emph{observables}. We will fist provide the general definition of observables in CTMCs below, followed by a more specific characterisation in stochastic rewriting systems as part of Theorem~\ref{thm:smf}.
\begin{defi}\label{def:obs}
Let $\cO_{\bfC}\subset End_{\bR}(\cS_{\bfC})$ denote the space of \emph{observables}, defined as the space of \emph{diagonal operators},
\begin{equation}
\cO_{\bfC}:=\{O\in End_{\bR}(\cS_{\bfC})\mid \forall o\in \obj{\bfC}_{\cong}:\; O\ket{o}=\omega_O(o)\ket{o}\,,\; \omega_O(o)\in \bR\}\,.
\end{equation}
We furthermore define\footnote{On distributions in $\cS_{\bfC}$ whose sum of coefficients is not finite, we consider $\bra{}$ to be undefined; in practice, we will however only be interested to evaluate moments of observables on (sub-)probability distributions, i.e.\ one may directly verify the finiteness properties in a given calculation.} the so-called \emph{projection operation} $\bra{}:\cS_{\bfC}\rightarrow \bR$ via extending by linearity the definition of $\bra{}$ acting on basis vectors of $\hat{\bfC}$, 
\begin{equation}
\forall o\in \obj{\bfC}_{\cong}:\quad \braket{}{o}:=1_{\bR}\,.
\end{equation}
These definitions induce a notion of \emph{correlators} of observables, defined for $O_1,\dotsc,O_n\in \cO_{\bfC}$ and $\ket{\Psi}\in \Prob{\bfC}$ as
\begin{equation}
\langle O_1,\dotsc,O_n\rangle_{\ket{\Psi}}:=\bra{}O_1,\dotsc,O_n\ket{\Psi}
=\sum_{o\in \obj{\bfC}_{\cong}}\psi_o\cdot\omega_{O_1}(o)\cdots \omega_{O_n}(o)\,.
\end{equation}
\end{defi}

The precise relationship between the notions of CTMCs and DPO rewriting rules as encoded in the rule algebra formalism is established in the form of the following theorem (compare~\cite{bdg2016}):

\begin{thm}[DPO-type stochastic mechanics framework]\label{thm:smf}
Let $\bfC$ be an $\cM$-adhesive category satisfying Assumption~\ref{ass:RAdpo}, and which in addition possesses an $\cM$-initial object. Let ${\{(\grule{O_j}{r_j}{I_j})\in \cR_{\bfC}\}}_{j\in \cJ}$ be a (finite) set of rule algebra elements, and ${\{\kappa_j\in \bR_{\geq 0}\}}_{j\in \cJ}$ a collection of non-zero parameters (called \emph{base rates}). Then one may construct a Hamiltonian $H$ from this data according to
\begin{equation}
H:=\hat{H}+\bar{H}\,,\quad
\hat{H}:=\sum_{j\in \cJ}\kappa_j\cdot \rho_{\bfC}\left(\grule{O_j}{r_j}{I_j}\right)\,,\quad
\bar{H}:=-\sum_{j\in \cJ}\kappa_j\cdot \rho_{\bfC}\left(\grule{I_j}{id_{dom(r_j)}}{I_j}\right)\,.
\end{equation}
Here, we define for arbitrary $(\GRule{O}{r}{I})\equiv(O\xleftarrow{o}K\xrightarrow{i}I)\in \Lin{\bfC}$
\begin{equation}
(\GRule{I}{id_{dom(r)}}{I}):=(I\xleftarrow{i}K\xrightarrow{i}I)\,.
\end{equation}
The \emph{observables} for the resulting CTMC are operators of the form
\begin{equation}
O_M^t=\rho_{\bfC}\left(\grule{M}{t}{M}\right)\,.
\end{equation}
We furthermore have the \emph{jump-closure property}, whereby for all $(\grule{O}{r}{I})\in \cR_{\bfC}$
\begin{equation}\label{eq:ojc}
\bra{}\rho_{\bfC}(\grule{O}{r}{I})=\bra{}O_I^{id_{dom(r)}}\,.
\end{equation}
\end{thm}
\begin{proof}
By definition, the DPO-type canonical representation of a generic rule algebra element $(\grule{O}{r}{I})\in \cR_{\bfC}$ is a row-finite linear operator, since for every $C\in \obj{\bfC}_{\cong}$ the set of admissible matches $\Match{p}{C}$ of the associated linear rule $p\equiv(\GRule{I}{r}{O})$ is finite. Consequently, $\rho_{\bfC}(\grule{O}{r}{I})$ lifts consistently from a linear operator in $End(\hat{\bfC})$ to a linear operator in $End(\cS_{\bfC})$. Let us prove next the claim on the precise structure of observables. Recall that according to Definition~\ref{def:obs}, an observable $O\in \cO_{\bfC}$ must be a linear operator in $End(\cS_{\bfC})$ that acts \emph{diagonally} on basis states $\ket{C}$ (for $C\in \obj{\bfC}_{\cong}$), whence that satisfies for all $C\in \obj{\bfC}_{\cong}$
\begin{equation*}
O\ket{C}=\omega_O(C)\ket{C}\quad (\omega_O(C)\in \bR)\,.
\end{equation*}
Comparing this equation to the definition of the DPO-type canonical representation (Definition~\ref{def:canRep}) of a generic rule algebra basis element $\delta(p)\in \cR_{\bfC}$ (for $p\equiv(I\xleftarrow{i}K\xrightarrow{o}O)\in \Lin{\bfC}$),
\begin{equation*}
\rho_{\bfC}(\delta(p))\ket{C}:=\begin{cases}
\sum_{m\in \Match{p}{C}}\ket{p_m(C)}\quad &\text{if }\Match{p}{C}\neq \emptyset\\
0_{\hat{\bfC}}&\text{else,}
\end{cases}
\end{equation*}
we find that in order for $\rho_{\bfC}(\delta(p))$ to be diagonal we must have
\begin{equation*}
 \forall C\in \obj{\bfC}:\forall m\in \Match{p}{C}:\quad p_m(C)\cong C\,.
\end{equation*}
But by definition of derivations of objects along admissible matches (Definition~\ref{def:DPOr}), the only linear rules $p\in \Lin{\bfC}$ that have this special property are precisely the rules of the form
\begin{equation*}
p^r_M= (M\xleftarrow{r}K\xrightarrow{r}M)\,.
\end{equation*}
In particular, defining $O_M^r:=\rho_{\bfC}(\delta(p^r_M))$, we find that the eigenvalue $\omega_{O^r_M}(C)$ coincides with the cardinality of the set $\Match{p^r_M}{C}$ of admissible matches,
\begin{equation*}
\forall C\in \obj{\bfC}_{\cong}:\quad  O_M^r\ket{C}=|\Match{p}{C}|\cdot\ket{C}\,.
\end{equation*}
This proves that the operators $O^r_M$ form a basis of diagonal operators on $End(\bfC)$ (and thus on $End(\cS_{\bfC})$) that arise from linear combinations of canonical representations of rule algebra elements.

To prove the jump-closure property, note that it follows from Definition~\ref{def:DPOr} that for an arbitrary linear rule $p\equiv(I\xleftarrow{i}K\xrightarrow{o}O)\in \Lin{\bfC}$, a generic object $C\in \obj{\bfC}$ and a $\cM$-morphism $m:I\rightarrow C$, the admissibility of $m$ as a match is determined by whether or not the match fulfils the gluing condition (Definition~\ref{def:gc}), i.e.\ whether or not the following pushout complement exists,
\begin{equation*}\gdef\mycdScale{0.85}
\begin{mycd}
I\ar[r,leftarrow,"i"]\ar[d,"m"'] & K\ar[d,"g",dashed]\ar[dl,phantom,"\mathsf{POC}"]\\
C\ar[r,leftarrow,dashed,"v"'] & E
\end{mycd}\,.
\end{equation*}
Thus we find that with $p'=(I\xleftarrow{i}K\xrightarrow{i}I)\in \Lin{\bfC}$, the set $\Match{p}{C}$ of admissible matches of $p$ in $C$ and $\Match{p'}{C}$ of $p'$ in $C$ have the same cardinality. Combining this with the definition of the projection operator $\bra{}$ (Definition~\ref{def:obs}),
\begin{equation*}
\forall C\in \obj{\bfC}_{\cong}:\quad \braket{}{C}:=1_{\bR}\,,
\end{equation*}
we may prove the claim of the jump-closure property via verifying it on arbitrary basis elements (with notations as above):
\begin{equation*}
\bra{}\rho_{\bfC}(\delta(p))\ket{C}=|\Match{p}{C}|=|\Match{p'}{C}|=\bra{}\rho_{\bfC}(\delta(p'))\ket{C}\,.
\end{equation*}
Since $C\in \obj{\bfC}_{\cong}$ was chosen arbitrarily, we thus have indeed that
\begin{equation*}
\bra{}\rho_{\bfC}(\delta(p))=\bra{}\rho_{\bfC}(\delta(p'))\,.
\end{equation*}
Finally, combining all of these findings, one may verify that $H$ as stated in the theorem fulfils all required properties in order to qualify as an infinitesimal generator of a continuous-time Markov chain.
\end{proof}

We illustrate the framework with an example of a stochastic rewriting system based on the category $\mathbf{uGraph}$ of finite undirected multigraphs and morphisms thereof, where we pick the two rule algebra elements $e_{+}$ and $e_{-}$ specified in~\eqref{eq:Adef} to define the transitions of the system. Together with two non-negative real parameters $\kappa_{+},\kappa_{-}\in \bR_{\geq0}$, the resulting Hamiltonian $H=\hat{H}+\bar{H}$ reads (with $E_{\pm}:=\rho(e_{\pm})$ and $O_{\bullet}$ as in~\eqref{eq:defOv})
\begin{equation}
\hat{H}=\kappa_{+}E_{+}+\kappa_{-} E_{-}\,,\quad \bar{H}=-\tfrac{1}{2}\kappa_{+}O_{\bullet}(O_{\bullet}-1)-\kappa_{-}O_E\,,
\quad O_E:=\tfrac{1}{2}\rho\left(\tP{%
  \vI{1}{1}{black}{}{black}{}
  \vI{1}{2}{black}{}{black}{}
  \eI{1}{1}{=}{black}{}{black}{}}\right)\,.
\end{equation}
Let us assume for simplicity that we start our evolution from an initial state $\ket{\Psi(0)}=\ket{G_0}$, with $G_0$ (the isomorphism class of) some finite undirected graph. We denote by $N_V$ and $N_E$ the number of vertices and edges of $G_0$, respectively, which may be computed as
\begin{equation}
\bra{}O_{\bullet}\ket{G_0}=N_V\,,\quad \bra{}O_E\ket{G_0}=N_E\,.
\end{equation}
Since the two linear rules that define the system create and delete edges, but do not modify the number of vertices, the time-dependent probability distribution $\ket{\Psi(t)}$ (for $t\geq 0$) with $\ket{\Psi(0)}=\ket{G_0}$ is supported on graph states that all have the same number of vertices $N_V$ as the initial graph $G_0$, which entails that
\begin{equation}\label{eq:OVstatic}
  \forall\; t\geq 0:\quad \bra{}O_{\bullet}\ket{\Psi(t)}=N_V\,.
\end{equation}
Let us thus focus on the dynamics of the edge-counting observable $O_E$. We follow the strategy put forward in~\cite{bdg2016,bdg2018} and consider the \emph{exponential moment generating function} $E(t;\epsilon)$ of $O_E$, defined as
\begin{equation}
E(t;\epsilon):=\bra{}e^{\epsilon O_E}\ket{\Psi(t)}\,,
\end{equation}
where $\epsilon$ is a formal variable. More explicitly, $E(t;\epsilon)$ encodes the \emph{statistical moments} of $O_E$, in the sense that for all (finite) $n\geq1$,
\begin{equation}
  \left[\tfrac{\partial^n}{\partial \varepsilon^n}E(t;\varepsilon)\right]\big\vert_{\varepsilon\to0}=\bra{}O_E^n\ket{\Psi(t)}\,.
\end{equation}
We may calculate the \emph{evolution equation} for $E(t;\varepsilon)$ as follows (compare~\cite{bdg2016,bdg2018}):
\begin{equation}\label{eq:appExAux1}
\begin{aligned}
  \tfrac{\partial}{\partial t}E(t;\varepsilon)&=
  \bra{}e^{\varepsilon O_E}H\ket{\Psi(t)}\\
  &=\bra{}\left(e^{\varepsilon O_E}He^{-\varepsilon O_E} \right)e^{\varepsilon O_E}\ket{\Psi(t)}\\
  &=\sum_{n\geq 0}\frac{1}{n!}\bra{}\left(ad_{\varepsilon O_E}^{\circ\:n}(H)\right)e^{\varepsilon O_E}\ket{\Psi(t)}\,.
\end{aligned}
\end{equation}
Here, in the last step, we have taken advantage of a variant of the BCH formula (see e.g.~\cite[Prop.~3.35]{hall2015lieGroups}), whereby for two composable linear operators $A$ and $B$ and for a formal variable $\lambda$,
\begin{equation}
  e^{\lambda A}Be^{-\lambda A}=\sum_{n\geq 0}\frac{1}{n!} ad_{\lambda A}^{\circ \:n}(B)\,,
\end{equation}
with the adjoint action defined via the so-called \emph{commutator} $[.,.]$,
\begin{equation}
  ad_A(B):=[A,B]=AB-BA \,.
\end{equation}
Moreover, we let $ad_A^{0}(B):= B$, and for $n\geq1$,
\begin{equation}
  ad_A^{\circ \:n}(B):=[A,[A,[\dotsc,[A,B]\dotsc]]]
\end{equation}
denotes the $n$-fold nested commutator. Taking advantage of the general fact that a Hamiltonian as constructed according to Theorem~\ref{thm:smf} verifies
\begin{equation}\label{eq:Haux}
\bra{}H=0\,,
\end{equation}
we may conclude that the term in~\eqref{eq:appExAux1} for $n=0$ vanishes identically. In order to compute the terms for $n\geq 1$, it is straightforward to verify that
\begin{equation}
  ad_{\varepsilon O_E}(\kappa_{+}E_{+})=\varepsilon\kappa_{+}E_{+}\,,\quad
  ad_{\varepsilon O_E}(\kappa_{-}E_{-})=-\varepsilon\kappa_{-}E_{-}\,,\quad
  ad_{\varepsilon O_E}(\bar{H})=0\,,
\end{equation}
which entails that\footnote{It may be worth emphasising that it is this particular type of calculation for which the rule algebra framework provides the technical prerequisites, as it would be otherwise impossible to reason about infinite series of causal interactions and rewriting steps.}
\begin{equation}
  \sum_{n\geq 1}\frac{1}{n!}\, ad_{\varepsilon O_E}^{\circ\:n}(H)
  =\kappa_{+}\left(e^{\varepsilon}-1\right)E_{+}
  +\kappa_{-}\left(e^{-\varepsilon}-1\right)E_{-}\,.
\end{equation}
To proceed, we invoke the \emph{jump-closure property} as described in~\eqref{eq:ojc} to conclude that
\begin{equation}\label{eq:appExJC}
  \bra{}E_{+}=\tfrac{1}{2}\bra{}O_{\bullet}(O_{\bullet}-1)\,,\quad
  \bra{}E_{-}=\bra{}O_E\,.
\end{equation}
Recalling our earlier result as presented in~\eqref{eq:OVstatic}, we find that
\begin{equation}
\tfrac{1}{2}\bra{}O_{\bullet}(O_{\bullet}-1)e^{\varepsilon O_E}\ket{\Psi(t)}
=\tfrac{1}{2}\bra{}e^{\varepsilon O_E}O_{\bullet}(O_{\bullet}-1)\ket{\Psi(t)}
\overset{\eqref{eq:OVstatic}}{=}\binom{N_V}{2} E(t;\varepsilon)\,,
\end{equation}
where in the first step we have made use of the fact that observables commute (i.e.\ in particular $[O_{\bullet},O_E]=0$). As for the contribution due to $\bra{}E_{-}$, note that
\begin{equation}
  \bra{}E_{-}e^{\varepsilon O_E}\ket{\Psi(t)}
  \overset{\eqref{eq:appExJC}}{=}\bra{}O_E e^{\varepsilon O_E}\ket{\Psi(t)}
  =\tfrac{\partial}{\partial \varepsilon} \bra{}e^{\varepsilon O_E}\ket{\Psi(t)}
  =\tfrac{\partial}{\partial \varepsilon} E(t;\varepsilon)\,.
\end{equation}
Assembling these results into~\eqref{eq:appExAux1}, and with $E(0;\varepsilon)=e^{\varepsilon N_E}$ (for $\ket{\Psi(0)}=\ket{G_0}$, and with $N_E$ edges in $G_0$), we obtain the following refined form for the evolution equation of $E(t;\varepsilon)$:
\begin{equation}\label{eq:appExAux2}
\begin{aligned}
  \tfrac{\partial}{\partial t}E(t;\varepsilon)&=
  \left[
    \kappa_{+}\binom{N_V}{2}\left(e^{\varepsilon}-1\right)
    +\kappa_{-}\left(e^{-\varepsilon}-1\right)\tfrac{\partial}{\partial\varepsilon}
  \right] E(t;\varepsilon)\,,\qquad E(0;\varepsilon)=e^{\varepsilon N_E}\,.
\end{aligned}
\end{equation}
In other words, we have thus transformed the problem of studying the dynamics of the edge-counting observable $O_E$ into the problem of studying the evolution-equation~\eqref{eq:appExAux2}. We may employ a standard technique well-known from the combinatorics literature to solve this problem in closed form, namely the so-called \emph{semi-linear normal-ordering technique} as introduced in~\cite{Dattoli:1997iz,blasiak2005boson,blasiak2011combinatorial} (and recently applied in~\cite{bdp2017} to semi-linear PDEs for chemical reaction systems). More concretely, we recognise that the differential operator in~\eqref{eq:appExAux2} has the ``semi-linear'' structure,
\begin{equation}
  h=q(\varepsilon)\tfrac{\partial}{\partial\varepsilon}+v(\varepsilon)\,,
  \quad
  q(\varepsilon)=\kappa_{-}\left(e^{-\varepsilon}-1\right)\,,\quad
  v(\varepsilon)=\kappa_{+}\binom{N_V}{2}\left(e^{\varepsilon}-1\right)\,.
\end{equation}
The general semi-linear normal-ordering formula then implies that given such a semi-linear differential operator and an evolution equation such as~\eqref{eq:appExAux2},
\begin{equation}
  \tfrac{\partial}{\partial t}E(t;\varepsilon)=\left[q(\varepsilon)\tfrac{\partial}{\partial\varepsilon}+v(\varepsilon)\right]E(t;\varepsilon)\,,\quad E(0;t)=E_0(\varepsilon)\,,
\end{equation}
the solution of this equation reads\footnote{To be fully precise, in a given problem one first has to compute the \emph{formal} solution (i.e.\ with $t$ a formal rather than a real-valued variable) using the normal-ordering formula, and then in a separate step verify that the solution thus obtained is convergent upon specialising $t$ to a real-valued variable.}
\begin{equation}\label{eq:evoResult}
\begin{aligned}
  E(t;\varepsilon)&=g(t;\varepsilon)E_0(T(t;\varepsilon))\,,\qquad
  \left\{
  \begin{array}{rcl}
  \tfrac{\partial}{\partial t}T(t;\varepsilon)&=&q(T(t;\varepsilon))\,,\quad T(0;\varepsilon)=\varepsilon\\
  \ln(g(t;\varepsilon))&=&\int_0^t dw\, v(T(w;\varepsilon))\,.
  \end{array}\right.
\end{aligned}
\end{equation}
Thus via solving the above PDE for $T(t;\varepsilon)$ and performing the integration to obtain $\ln(g(t;\varepsilon))$, we finally arrive at the following closed-form solution of the evolution equation~\eqref{eq:appExAux2}:
\begin{equation}
E(t;\varepsilon)=e^{\frac{\kappa_{+}}{\kappa_{-}}\binom{N_V}{2}\left(e^{\varepsilon}-1\right)\left(1-e^{-\kappa_{-}t}\right)}{\left(\left(e^{\varepsilon}-1\right)e^{-\kappa_{-}t}+1\right)}^{N_E}\,.
\end{equation}
For illustration, we present in Figure~\ref{fig:timeEv} the time-evolution of $\langle O_E\rangle(t)$ (whence of the first $\varepsilon$-derivative of $E(t;\varepsilon)$ evaluated at $\varepsilon=0$) for three different choices of parameters $\kappa_{+}$ and $\kappa_{-}$, and for four different choices each of initial number of edges $N_E$.

\medskip
As a further refinement, since $E(t;\varepsilon)$ is the moment-generating function of a univariate probability distribution, we may take advantage of the well-known relationship (see e.g.\ \cite{bdp2017} for further details) between the moment-generating function $E(t;\varepsilon)$ and the \emph{probability generating function (PGF)} $P(t;\lambda)$,
\begin{equation}
P(t;\lambda)=\sum_{n\geq 0}p_n(t)\lambda^n=E(t;\ln\lambda)\,,
\end{equation}
with $p_n(t)$ interpreted as the probability to count precisely $n$ edges at time $t$ (for $t\geq0$). Thus we may transform the result~\eqref{eq:evoResult} into the easier to interpret form
\begin{equation}
\begin{aligned}
  P(t;\lambda)&=Pois\left(\lambda;\frac{\kappa_{+}}{\kappa_{-}}\binom{N_V}{2}\left(1-e^{-\kappa_{-}t}\right)\right)Binom(\lambda;e^{-\kappa_{-}t},N_E)\\
  Pois(\lambda;\alpha)&=e^{\alpha(\lambda-1)}\,,\quad
  Binom(\lambda;\alpha,N)={\left(\alpha\lambda +(1-\alpha)\right)}^N\,,
\end{aligned}
\end{equation}
where $Pois(\lambda;\alpha)$ denotes the PGF of \emph{Poisson distribution} (of parameter $0\leq \alpha <\infty$), and where $Binom(\lambda;\alpha,N)$ denotes the PGF of a \emph{Binomial distribution} (of parameters $0\leq \alpha\leq 1$ and $N\in\mathbb{Z}_{\geq 0}$). Referring yet again to~\cite{bdp2017} for further details, since the PGF of the \emph{convolution} of two probability distributions is given by the product of their PGFs, we thus find that the dynamics of the edge-counting observable $O_E$ is described in terms of a \emph{convolution} of a Poisson-distribution with a binomial distribution. Moreover, in the limit $t\to\infty$ we simply find that the number of edges in the distribution over graph states is Poisson-distributed,
\begin{equation}
\lim\limits_{t\to\infty}P(t;\lambda)=Pois\left(\lambda;\frac{\kappa_{+}}{\kappa_{-}}\binom{N_V}{2}\right)\,.
\end{equation}
Interestingly, the coefficient $\binom{N_V}{2}$ in this equation is precisely the number of edges of a \emph{complete graph} on $N_V$ vertices. Another interesting observation concerns a special choice of base rates $\kappa_{\pm}$ and initial state $\ket{\Psi(0)}$: if $\kappa_{+}=\kappa_{-}$ and $N_E=N_{E*}=\binom{N_V}{2}$, one may compute from~\eqref{eq:evoResult} $\langle O_E\rangle(t)=N_{E*}=const$ for all $t\geq 0$. All of these findings combined entail that the edge creation and deletion process described here is in fact nothing else but a so-called \emph{birth-death process} of random deletion and creation of ``particles'', with the role of ``particles'' played in the present case by the edges of the graphs that the system evolves upon. This result might be somewhat anticipated, in that for the special case $N_V=2$ we found in the previous section that $E_{+}$ and $E_{-}$ acting on the states with two vertices effectively yield a representation of the Heisenberg-Weyl algebra, whence in this case the process reduces trivially to a birth-death process on edges with rates $\kappa_{+}$ and $\kappa_{-}$ (see~\cite{bdp2017} for further details on chemical reaction systems). As an outlook, we have conducted in~\cite{bdg2018} a full study of the interesting phenomenon of a stochastic rewriting system on state-spaces of graph-like structures exhibiting dynamics that is comparable in nature and mathematical structure to the dynamics of discrete transition systems such as chemical reaction systems and branching processes.

\begin{figure}
  \caption{Time-evolution of $\langle O_E\rangle(t)$ for $\ket{\Psi(0)}=\ket{G_0}$ with $N_V=100$.\label{fig:timeEv}}
  \centering
    \includegraphics[width=0.5\textwidth]{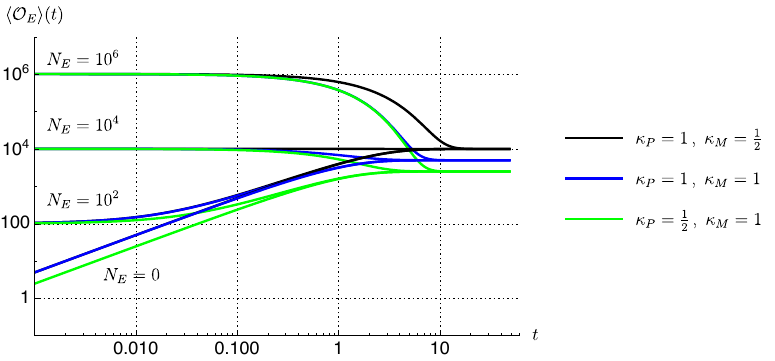}
\end{figure}

\section{Conclusion and Outlook}%
\label{sec:conclusion}

The elegance and effectiveness of traditional Double-Pushout (DPO) rewriting over $\cM$-adhesive categories~\cite{ehrig:2006aa} consists in describing in one uniform categorical framework a wide variety of rewriting systems (cf.\ e.g.\  Table~\ref{tab:adh}) of relevance in many practical applications. It must therefore be considered as somewhat of a striking surprise that techniques of rewriting are not more ubiquitous in the applied mathematics, combinatorics and life sciences literature, given that in many of these fields, models and systems of computations formulated in terms of manipulations of graph-like structures are very frequently utilised. Originally motivated~\cite{bdp2017,bdg2016} by studying problems in the setting of chemical reaction systems and in stochastic graph rewriting systems, we introduce in this paper what we believe to constitute a first essential stepping stone in order to overcome the aforementioned conceptual divide: the \emph{DPO rule algebra} framework. Intuitively, an individual DPO rewriting rule may typically be applied in a number of different ways (i.e.\ mediated by different admissible matches) to a given object; the key idea of the present paper consists in encoding this particular form of \emph{non-determinism} in a setting of vector spaces (with basis indexed by isomorphism classes of objects) and linear operators mapping a given basis vector to ``the sum over all possible applications'' of a given rule. Modulo certain technical implementation details, this approach however immediately raises a conceptual question about the precise nature of \emph{sequential applications} of rules to objects: it must be ensured that applying two particular rules in sequence to a given object ``in all possible ways'' is equivalent to applying the aforementioned linear operators (one after the other) onto the basis vector indexed by the object. Inspired by earlier work~\cite{bdg2016} on multigraph DPO rewriting systems implemented in a formulation based upon binary relations, we present here a general implementation of a consistent mathematical framework that reveals the notion of an \emph{associative unital algebra} that may be constructed based upon sequential compositions of DPO rules. As illustrated in the special case of DPO rewriting systems on \emph{discrete graphs} (Section~\ref{sec:HW}), defining a \emph{representation} for the aforementioned algebra amounts to precisely the construction of certain linear operators that encode faithfully the sequential applications of DPO rules to objects. Our general construction hinges upon our novel theorem on the associativity of the operation of forming DPO-type concurrent compositions of linear rewriting rules (Theorem~\ref{thm:assocDPO}), based upon which we introduce the concept of \emph{rule algebras} in Definition~\ref{def:RADPO}: each (isomorphism class of a) linear rule is mapped to an element of an abstract vector space, on which the sequential rule composition operation is implemented as a certain bilinear binary operation. For every $\cM$-adhesive category $\bfC$ that is finitary and that possesses $\cM$-effective unions, the associated rule algebra is associative, and if the category possesses an $\cM$-initial object, this algebra is in addition unital. Each rule algebra element in turn gives rise via the construction of the \emph{canonical representation} (Definition~\ref{def:canRep}) to a linear operator acting upon the vector space with basis indexed by isomorphism classes of objects, solving precisely the problem of encoding the non-determinism in rule applications in a principled and general fashion. We then hinted at the potential of our approach in the realm of combinatorics in Section~\ref{sec:RT}, and, as a first major application of our framework, we presented in Section~\ref{sec:SM} a \emph{universal construction of continuous-time Markov chains} based upon linear DPO rules of $\cM$-adhesive categories $\bfC$.

The general motivation of this paper consisted in rendering techniques of rewriting more accessible to practitioners and theoreticians beyond the traditional rewriting theory community. In line with these efforts, we have recently continued the development of our rule algebra framework to include another key framework for categorical rewriting over adhesive categories (so-called \emph{Sesqui-Pushout (SqPO) rewriting}) in~\cite{Behr_2019_sqpo}, and a variant for both DPO- and SqPO-type rewriting in the settings of rewriting with \emph{application conditions} and with \emph{constraints} on objects in~\cite{behr2019compositionality}. The latter contains a general associativity theorem for rules with conditions, thus providing the  prerequisites for achieving a general rule algebra framework (work currently in progress) that can capture the graph-like data structures of relevance in many practical applications that differ from idealised mathematical structures such as directed multigraphs via certain additional constraints on objects (such as e.g.\ in planar binary trees). In terms of analysing stochastic rewriting systems via rule-algebraic methods, we have recently reported in~\cite{bdg2018} a first set of custom techniques that permit to exploit the rule algebra structure in order to extract dynamical information from these stochastic systems. Finally, the associativity property of rule compositions itself has been demonstrated in~\cite{behr2019tracelets} to give rise to a novel concept of so-called \emph{tracelets}, which permit to minimally and fully characterise the compositional and causal properties underlying sequences of applications of rewriting rules (so-called derivation traces). While currently still in its early stages of development, we envision that our novel viewpoint on the study of rewriting systems will ultimately yield rich and fruitful interactions of the specialist field of categorical rewriting theory with the broader applied mathematics, computer and life sciences research fields.

\end{document}